\documentclass[journal]{IEEEtran}
\usepackage{graphicx}
\usepackage{cite}
\usepackage{epstopdf}

\usepackage[ruled,vlined,linesnumbered]{algorithm2e}
\usepackage[cmex10]{amsmath}
\usepackage{amssymb}
\usepackage{amsthm}
\usepackage{graphicx}
\usepackage{caption}
\usepackage{verbatim}
\usepackage{mathrsfs}
\usepackage{stfloats}
\usepackage{color}
\usepackage{mathrsfs}
\usepackage{cases}
\usepackage{makecell}
\usepackage{bm}

\allowdisplaybreaks[4]

\newtheorem{theorem}{Theorem}
\newtheorem{lemma}{Lemma}

\newtheorem{definition}{Definition}

\begin{document}

\title{User Pairing and Power Allocation in Untrusted Multiuser NOMA for Internet-of-Things}

\author{\IEEEauthorblockN{Chaoying Yuan, Wei~Ni,~\IEEEmembership{Senior~Member,~IEEE}, Kezhong Zhang, Jingpeng Bai, Jun Shen, \\
and Abbas Jamalipour,~\IEEEmembership{Fellow,~IEEE}
\vspace{-1em}}

\thanks{C. Yuan is with the China Telecom Corporation Limited, Shanghai, China, 200122 (e-mail: yuancy3@chinatelecom.cn).

W. Ni is with the Data61, CSIRO, Marsfield, NSW 2122, Australia (e-mail: wei.ni@data61.csiro.au).

K. Zhang graduated from the Beijing University of Posts and Telecommunications (e-mail: zhangkz@foxmail.com).

J. Bai and J. Shen are with the China Telecom Corporation Limited, Guangzhou, China, 510630 (e-mail: baijp@chinatelecom.cn; shenjun6@chinatelecom.cn).

A. Jamalipour is with the School of Electrical and Information Engineering, The University of Sydney, NSW 2006, Australia (email: a.jamalipour@ieee.org).}

}

\maketitle

\begin{abstract}
In the Internet-of-Things (IoT), massive sensitive and confidential information is transmitted wirelessly, making security a serious concern. This is particularly true when technologies, such as non-orthogonal multiple access (NOMA), are used, making it possible for users to access each other's data. This paper studies secure communications in multiuser NOMA downlink systems, where each user is potentially an eavesdropper. Resource allocation is formulated to achieve the maximum sum secrecy rate, meanwhile satisfying the users' data requirements and power constraint. We solve this non-trivial, mixed-integer non-linear programming problem by decomposing it into power allocation with a closed-form solution, and user pairing obtained effectively using linear programming relaxation and barrier algorithm. These subproblems are solved iteratively until convergence, with the convergence rate rigorously analyzed. Simulations demonstrate that our approach outperforms its existing alternatives significantly in the sum secrecy rate and computational complexity.
\end{abstract}
\IEEEpeerreviewmaketitle

\section{Introduction}
\IEEEPARstart{T}{he}
%more and more Internet of Things (IoT) devices are connected to the wireless network, IoT has become the mainstream communication paradigm for connecting the Internet to everyday physical things~\cite{Zhang2022}. Cellular networks become the primary access network for IoT connections. Ericsson predicts that around 5.9 billion cellular IoT devices will be deployed in 2026\cite{Vaezi2022}. These devices collect data, process it, and make intelligent decisions, changing the world around us for the better. While IoT has brought great convenience to our daily lives, it has also brought some thorny challenges. Due to the extreme scarcity of wireless resources, today's networks are increasingly unable to provide large-scale connectivity and satisfactory wireless communication. On the other hand, with a large number of devices entering our daily life and the industrial Internet, ensuring the security and privacy of the IoT is also a key requirement.
increasing number of Internet-of-Things (IoT) devices connected to wireless networks has made IoT the dominant communication paradigm for connecting the physical world to the Internet~\cite{Zhang2022}.
%Cellular networks have become the primary access network for IoT connections, with
Ericsson predicted that around 5.9 billion cellular IoT devices will be deployed by 2026~\cite{Vaezi2022}. These devices collect and process data, and make intelligent decisions, improving efficiency, productivity, and convenience. However, the rapid expansion of IoT has brought challenges, including wireless resource scarcity and security and privacy concerns. It is crucial to address these issues in order to ensure satisfactory wireless communication and protect the security and privacy of IoT devices.

In the IoT scenarios, the adoption of non-orthogonal multiple access (NOMA) could potentially improve the connectivity and efficiency of massive IoT devices~\cite{Islam2017}. In contrast to orthogonal multiple access (OMA) in the time, frequency, and code domains, a NOMA transmitter can allocate different transmit powers for different receivers within the same resource block, according to the channel conditions of the receivers. Supposition coding (SC) is adopted at the transmitter. Successive interference cancellation (SIC) is deployed at the receivers. However, the broadcast nature of radio and the use of SIC at the receivers make NOMA susceptible to attacks launched by external and internal eavesdroppers~\cite{Hamamreh2019}.

To address these security concerns, physical layer security (PLS) techniques have been considered a promising approach~\cite{Hamamreh2019}. PLS
%uses the inherent properties of the physical layer of communication systems, such as the channel characteristics, to provide secure communication. Moreover, PLS schemes
can be computationally effective compared to other forms of security, such as cryptography, because they rely on simple operations that can be performed at the physical layer, such as power allocation~\cite{Wang2019}. PLS is more appropriate for low-cost IoT devices that often have limited computing resources and energy constraints~\cite{Wang2019a}. By using simple and efficient PLS techniques, IoT devices can achieve strong confidentiality of their communications without incurring high computational or energy costs~\cite{Mukherjee2015}.
Some other recent studies, e.g.,~\cite{Zhang2022,Ruby2022,Xiang2020}, also attempted to improve the secrecy performance of NOMA-based IoT systems in the presence of external eavesdropping.

It is possible for some users to eavesdrop on the signals intended for other users by executing the SIC, since users in a NOMA system share the same resource block.
%This means that at least some level of secrecy must be provided against internal eavesdropping as well.
Most studies have been under a two-user setting:
Some assumed far users untrusted~\cite{ElHalawany2018,Cao2020a}, and others assumed near users untrusted~\cite{Zhang2020,Cao2020}. Several studies~\cite{Thapar2020,Hota2022,Amin2022a} considered both users were untrusted. Different from the two-user settings in~\cite{Thapar2020,Hota2022,Amin2022a}, the
%security for untrusted multiuser scenarios is a realistic and challenging problem. The
authors of \cite{Thapar2021} proposed a decoding ordering criterion for untrusted multiuser NOMA with persistent power allocation. The authors considered all users share the same resource block, leading to fast growing interference and complexity at the receivers with the increase of users.
%In this case, the users can be potentially divided into multiple pairs, where NOMA is implemented within each pair.

In this paper, we investigate multiuser NOMA systems for IoT applications in the presence of untrusted users. User pairing and power allocation are optimized jointly to maximize the sum secrecy rate of the systems.
%However, a few challenges need to be addressed. First, one needs to properly assign the transmit power to different IoT devices in order to meet the requirements of the devices. Second, too many devices sharing the same resource block may lead to strong co-channel interference, which increases the complexity of the receiver and reduces the secrecy performance. To this end, a reasonable user pairing strategy is indispensable.}
To the best of our knowledge, user pairing and power allocation, which are critical to multiuser NOMA, have never been jointly considered in untrusted multiuser NOMA systems in the literature.

%No existing study has jointly considered user pairing and power allocation in untrusted multiuser NOMA systems.

%In this paper, we propose a lightweight, secure resource allocation among untrusted NOMA users.
The key contributions of this paper are:
\begin{itemize}
\item We study a new problem to maximize the sum secrecy rate of multiuser NOMA with  untrusted IoT users, by jointly optimizing power allocation and user pairing.

\item To effectively maximize the sum secrecy rate of all IoT devices, we adopt alternating optimization to circumvent the non-convexity of the new problem and decouple user pairing and power allocation.

\item Given user pairing, we derive analytically the optimal power allocation in closed form. Then, we develop user pairing obtained effectively using linear programming relaxation and the barrier method.

\item
Rigorous analyses are conducted for the convergence rate and complexity of our algorithm, confirming the validity of the algorithm.
\end{itemize}
Our approach addresses the challenges of interference and implementation complexity in untrusted multiuser NOMA systems, and has the potential to improve the security and performance of these systems. Extensive simulations demonstrate that the approach has superior secrecy performance, compared to existing schemes, and that joint consideration of user pairing and power allocation is critical for achieving this performance.

The remainder of the paper is arranged in the following way. In Section II,
the related works are reviewed. Section~III defines the system setting. In Section IV, the problem statement is provided, and the solution is delivered. In Section V, simulation results are analyzed to show the merits of our solution. Finally, this article is concluded in Section VI.

\textit{Notations:} Upper- and lower-case symbols stand for matrices and vectors, respectively; $^{T}$ denotes transpose; $\preccurlyeq$ stands for component-wise less than; $\cup$ and $\cap$ stand for the union and intersection operations, respectively; $\nabla$ denotes gradient. Tab. \ref{tab:notation_list} summarizes notations used in this paper.

\section{Related Work}
Most of the existing NOMA security studies have focused on external eavesdropping.
Secure transmissions in a NOMA-based IoT system were investigated in~\cite{Zhang2022}. The system offered the users different communication requirements. The authors of \cite{Ruby2022} studied the secrecy performance of cooperative NOMA-assisted IoT and derived the security outage probability under either a single- or multi-antenna setting. The authors of \cite{Xiang2020} jointly designed beamforming vector, power and subcarrier allocation to improve the worst-case sum secrecy rate in a multicarrier NOMA-assisted IoT system.

Another potential security threat in a NOMA system comes from internal users. A simple two-user setting has been actively studied.  ElHalawany \emph{et al}.~\cite{ElHalawany2018} studied the secrecy outage probability in a two-user NOMA system under the assumption that the far user was untrusted. In \cite{Cao2020a}, two optimal relay selection schemes were designed, and closed-form expressions of the secrecy outage probability was derived. Zhang \emph{et al}. \cite{Zhang2020} proposed an optimal decoding order of SIC and a jammer-aided cooperative jamming scheme for NOMA systems to defend against a stronger, near-user eavesdropper to improve the secrecy rate of the systems. In \cite{Cao2020}, a secure beamforming and power allocation strategy was designed to evaluate the secrecy outage probability of the systems in the presence of an untrusted near user.
Unlike \cite{ElHalawany2018,Cao2020a,Zhang2020,Cao2020}, the authors of \cite{Thapar2020,Hota2022,Amin2022a} treated both far and near users as the untrusted users. Specifically, the authors of \cite{Thapar2020} proposed an optimal decoding order to maximize secrecy fairness of a NOMA system. Hota \emph{et al}.~\cite{Hota2022} analyzed the ergodic rate and the ergodic secrecy rate of a two-user untrusted NOMA system with imperfect SIC. Amin \emph{et al}. \cite{Amin2022a} studied the secrecy rate maximization of the a trusted decode-and-forward relay-assisted NOMA system by optimizing power allocation. These designs provided secure communication by addressing the potential for internal eavesdropping in NOMA.

{
\begin{table}%[!htb]
	\caption{Notation list}
	\label{tab:notation_list}
\begin{center}
	\begin{tabular}{ c| l }
	 \hline
	 \multicolumn{1}{c|}
	 {\textbf{Notations}} &
	 \multicolumn{1}{c}
	 {\textbf{Descriptions}} \\
	 \hline
	 $\mathcal{D}$ & The disc-shaped area centered at the BS\\	
	 $s_k$ & The data symbol designed for user $k$ \\
	 $s$ & Transmit signal\\
	 $p_k$ &  The transmit power assigned for user $k$ \\
	 $g_k$ & The Rayleigh fading channel coefficient of user $k$ \\
	 $d_k$ & \makecell[l]{The distance between user $k$ and the BS} \\
	 $h_k$ & The channel impulse response of user $k$ \\
	 $\sigma$ & \makecell[l]{The standard deviation of the AWGN} \\
	 $\gamma_{m,n}$ & \makecell[l]{The SNR of user $m$ decoded by user $n$} \\
	 $w_{k}$ & \makecell[l]{The AWGN at user $k$} \\
	 $R^s_{n}$ & The secrecy rate of user $n$ \\
	 $R_{m, n}$ & \makecell[l]{The achievable rate of user $m$ decoded by user $n$} \\
	 $x_{m, n}$ & \makecell[l]{User pairing indicator; if user $m$ and $n$ share the same \\  resource block, $x_{m,n} = 1$. Otherwise, $x_{m, n} = 0$ }\\
	 $\mathbf{X}$ & \makecell[l]{The matrix of user pairing of which the $(m, n)$-th \\element is $x_{m,n}$}\\
	 $\mathrm{x}$ & \makecell[l]{The vectorization of the elements above the main \\diagonal of $\mathbf{X}$ in the row-major order} \\
	 $\mathbf{\hat{X}}$ & \makecell[l]{The continuous relaxation of $\mathbf{X}$ with the $(m, n)$-th element, \\ $\hat{x}_{m,n}\in[0, 1]$, indicating how likely user $m$ and $n$ are \\ paired to share a resource block.}\\
	 $\mathrm{\hat{x}}$ & \makecell[l]{The vectorization of the elements above the main \\diagonal of $\mathbf{\hat{X}}$ in row-major order} \\
	 $p_{m,n}$ & \makecell[l]{The transmit powers for paired users $m$ and $n$} \\
%	 $p_m$ & \makecell[l]{The transmit power for user $m$}\\
	 $\mathbf{p}$ & \makecell[l]{The vectorization of $\left(p_m+p_n\right)$ in row-major order}\\
	 $\mathbf{1}$ & A vector with all one entries\\
	 $\mathbf{0}$ & A vector with all zero entries\\
%	 $\preccurlyeq $ & \makecell[l]{Componentwise inequality between two vectors}\\
	 $\mathbf{w}$ & The Lagrange variable \\
	 $\mathbf{I}$ & The identify matrix\\
	 $\mathbf{A}$ &
	 $\mathbf{A}=\left[\mathbf{I}, -\mathbf{I}, \mathbf{p}\right]^T$\\
	 $\mathbf{K}\left(\mathbf{x}, \mathbf{w}\right)$ & The Karush-Kuhn-Tucker (KKT) matrix\\
%	 \hline
%	 $\mathbf{D}$ & \makecell[l]{the $n$-th row of $\mathbf{D}\in \mathbb{R}^{2K\times K(2K-1)}$, denoted by $\mathbf{d}_n^T$,\\ satisfies ${\mathbf{d}_n^T\hat{\mathbf{x} }= \sum_{m-1}^{2K}\hat{x}_{m, n}}$}\\
%	 \hline
%	 $\mathbf{J}\left( \mathbf{w}\right)$ & $\mathbf{J}\left( \mathbf{w}\right)=\left(\nabla g\left(\hat{\mathbf{x}}\right) + \mathbf{D}^T\mathbf{w}, \mathbf{D\hat{x}}-\mathbf{1}\right)$\\
	 $S$ & The upper bound of $\begin{Vmatrix}\mathbf{K}\left(\mathbf{x}, \mathbf{w}\right)^{\dagger} \end{Vmatrix}_F$\\
	 $L$ & \makecell[l]{The Lipschitz constant satisfying:
	 $\forall \left(\mathbf{x}_i, \mathbf{w}_i\right)$, \\ $\begin{Vmatrix} \left(\mathbf{x}_i, \mathbf{w}_i\right)\end{Vmatrix}_F
	 \leqslant
	 \begin{Vmatrix} \left(\mathbf{x}^{(0)}, \mathbf{w}^{(0)}\right)\end{Vmatrix}_F$}\\
	 $N$ & \makecell[l]{The number of iterations for user pairing to converge}\\
	 $\zeta$, $\tau$ & The control factors in backtracking line search\\
	 $\epsilon$ & The error tolerance of user pairing\\
	 $\xi$ & The control factor of the step size for user pairing \\
	 $\eta$ & The tolerance level of the overall algorithm\\
%	 $\mathcal{O}$ & Asymptotic notation\\
% 	%  $N_l$ & \makecell[l]{The number of iterations of the backtracking\\ line search in the damped Newton phase}\\
% 	%  \hline
% 	%  $N_\mathrm{N}$ & \makecell[l]{The number of iterations of the backtracking\\ line search}\\
% 	%  \hline
% 	%  $\mathcal{T}_1$ & \makecell[l]{The complexity of backtracking line search in the \\damped Newton phase in each iteration}\\
% 	%  \hline
% 	%  $\mathcal{T}_s$ & \makecell[l]{The complexity of backtracking line search in the \\damped Newton phase}\\
% 	%  \hline
% 	 $\mathcal{T}_{\mathbf{\hat{x}}}$ & \makecell[l]{The complexity of updating $\mathbf{\hat{x}}$}\\
% 	 $\mathcal{T}_{\mathbf{w}}$ & \makecell[l]{The complexity of updating $\mathbf{w}$}\\
% 	 $N_\mathrm{D}$ & The number of iterations of damped Newton phase\\
% 	 $\mathcal{T}_{\mathrm{D}}$ & The complexity of the damped Newton phase\\
% 	 $N_{\mathrm{Q}}$ & \makecell[l]{The number of iterations of the quadratically \\convergent phase}\\
% 	 $\mathcal{T}_{\mathrm{Q}}$ & The complexity of the quadratically convergent phase\\
% 	 $\mathcal{T}_{\mathrm{LP}}$ & The complexity of linear programming\\
% 	 $\mathcal{T}_g$ & \makecell[l]{The complexity of greedy based approach in user \\pairing algorithm}\\
% 	 $\mathcal{T}$ & The complexity of overall algorithm \\
% 	 $\mathrm{graph}\left(\cdot\right)$ & The graph of a mapping\\
	 $\delta_{\mathcal{F}}$ & The indicator function on the feasible domain $\mathcal{F}$\\
	 \hline
	\end{tabular}
	\end{center}
\end{table}
}

Compared with a two-user setting~\cite{Thapar2020,Hota2022,Amin2022a}, the security of a multi-user untrusted scenario is a more realistic and challenging problem. The most relevant, existing study \cite{Thapar2021} proposed a decoding order strategy for multi-user untrusted NOMA with fixed power allocation. However, excessive devices sharing the same resource block may lead to severe co-channel interference. To this end, an adequate user pairing strategy, in coupling with effective power allocation, is critical. In \cite{Zhou2017} and \cite{Ficken2015}, a Gale-Shapley algorithm-based and a Simplex method-based approaches were developed and dedicated to user pairing, respectively. Compared to the user pairing strategies developed in \cite{Zhou2017} and \cite{Ficken2015}, our approach delivers effective user pairing solution using logarithmic barrier method in couple with closed-form optimal power allocation, hence achieving improved efficiency and accuracy.

\section{System model}
In this paper, we investigate a multiuser downlink NOMA system with a base station (BS) and $2K$ untrusted users.
The users are untrusted in the sense that each user in the system may act as a potential eavesdropper and may attempt to intercept the confidential messages transmitted by other users to its own advantage.
The users are dispersed within a disc-shaped area $\mathcal{D}$ centered at the BS. The BS and users are equipped with omnidirectional antennas. The direct link between the BS and each user experiences Rayleigh fading~\cite{WangTCOM2015}. In order to reduce complexity, we divide the users into $K$ pairs, with each pair occupying a different resource block. This allows us to consider the system in manageable chunks and design efficient resource allocation strategies.

At the BS, the transmit signal for the users at each pair is
\begin{equation}
s = \sqrt {{p_m}} {s_m} + \sqrt {{p_n}} {s_n},
\end{equation}
where ${s_{k}}$ $\left( {k = m,n} \right)$ is the data symbol destined for user $k$ with unit energy ${\mathbb E}[|{s_{k}}|^2] = 1$, and ${p_{k}}$ represents the corresponding transmit power assigned for user.

The received signal of user $k$ is given by
\begin{align}
{y_k} = {h_k}\left( {\sqrt {{p_m}} {s_m} + \sqrt {{p_n}} {s_n}} \right) + {\omega _k},
\end{align}
where ${h_k} = {\text{ }}{{\text{g}}_k}d_k^{ - \alpha }$ with $g_k$ being the Rayleigh fading coefficient, ${d_k}$ the distance of user $k$ from the BS, and $\alpha $ the path loss;  ${w_{k}}$ is the zero-mean additive white Gaussian noise (AWGN) with variance ${\sigma ^2}$.

Assume the paired user with ${\left| {{h_m}} \right|^2} < {\left| {{h_n}} \right|^2}$. By following the NOMA principle, user $n$ with a higher channel gain first decodes the signal of user $m$, and then executes SIC to decode its own signal. User $m$ with poor channel gain first decodes its own signal and then executes SIC to decode user $n$'s signal. As such, we have
\begin{align}
&{\gamma _{m,n}} = \frac{{{p_m}{{\left| {{h_n}} \right|}^2}}}{{{p_n}{{\left| {{h_n}} \right|}^2} + {\sigma ^2}}},\hspace{+4mm} {\gamma _{n,n}} = \frac{{{p_n}{{\left| {{h_n}} \right|}^2}}}{{{\sigma ^2}}},\\
&{\gamma _{m,m}} = \frac{{{p_m}{{\left| {{h_m}} \right|}^2}}}{{{p_n}{{\left| {{h_m}} \right|}^2} + {\sigma ^2}}}, \hspace{+3mm}{\gamma _{n,m}} = \frac{{{p_n}{{\left| {{h_m}} \right|}^2}}}{{{\sigma ^2}}},
\end{align}
where ${\gamma _{m,n}}$ is the signal-to-interference-plus-noise-ratio (SINR)
of user $m$ decoded by user $n$, and ${\gamma _{n,m}}$ is the other way around.

Then, the achievable rates of the paired users are
\begin{align}
 & {R_{n,n}} = {\log _2}\left( {1 + {\gamma _{n,n}}} \right); \hfill \\
 & {R_{m,m}} = {\log _2}\left( {1 + {\gamma _{m,m}}} \right). \hfill
\end{align}

The secrecy rate ${R^s_{n}}$ of user $n$ is defined as
\begin{align}
{R^s_{n}} = \max \left\{ {{R_{n,n}} - {R_{n,m}},0} \right\}.
\end{align}
Here, ${R_{n,m}} = {\log _2}\left( {1 + {\gamma _{n,m}}} \right)$  is the eavesdropping rate of user $m$ on user $n$'s message. A positive secrecy rate can be awarded since ${\left| {{h_m}} \right|^2} < {\left| {{h_n}} \right|^2}$.

\section{Problem Statement and Proposed Solution}
In this section, power allocation and user pairing are optimized in an attempt to achieve the maximum sum secrecy rate under the data rate and  transmit power constraints.
%Let ${x_{m,n}} \in \left\{ {0,1} \right\}$, $1 \leqslant m,n \leqslant 2K$. If user $m$ is served together with user $n$, then ${x_{m,n}} = 1$. Otherwise, ${x_{m,n}} = 0$.
Let ${x_{m,n}} \in \left\{ {0,1} \right\}$ denote the binary scheduling variables. If user $m$ is served together with user $n$, we have ${x_{m,n}} = 1$. Otherwise, ${x_{m,n}} = 0$.  The considered problem is cast as
\begin{subequations}
\label{equ:original_problem}
\begin{align}
  \max\limits_{{x_{m,n}},{p_n},{p_m}} \hspace{+3mm}&\sum\limits_{m = 1}^{2K} {\sum\limits_{n = m + 1}^{2K}  } {{x_{m,n}}}{R^s_{n}} \hfill \\
 \text{s.t.} \hspace{+3.5mm}&{R_{m,m}} \geqslant {x_{m,n}}{R_m}, \hfill \\
 &{R_{n,n}} \geqslant {x_{m,n}}{R_n}, \hfill \\
 & \sum\nolimits_{m = 1}^{2K} {\sum\nolimits_{n = m + 1}^{2K} {{x_{m,n}}\left( {{p_n} + {p_m}} \right)} }  \leqslant P,\hfill \\
 &{x_{m,n}} \in \left\{ {0,1} \right\},1 \leqslant m,n \leqslant 2K, \hfill \\
& {x_{m,n}} = {x_{n,m}},1 \leqslant m,n \leqslant 2K, \hfill \\
&  \sum\nolimits_{m = 1}^{2K} {{x_{m,n}}}  = 1,1 \leqslant n \leqslant 2K, \hfill \\
&\sum\nolimits_{n = 1}^{2K} {{x_{m,n}}}  = 1,1 \leqslant m \leqslant 2K, \hfill
\end{align}
\end{subequations}
where $P$ is the total transmit power of the BS; ${R_m}$ and ${R_n}$ are the achievable rates of user $m$ and user $n$ in an OMA system, respectively, and
$${R_m}\! = \!\frac{1}{2}{\log _2}\left( \!\! {1 \!\! + \!\frac{{{p_{m,n}}{{\left| {{h_m}} \right|}^2}}}{{{\sigma ^2}}}} \!\!\right);{R_n} \!= \!\frac{1}{2}{\log _2}\left(\!\! {1 \!\!+\! \frac{{{p_{m,n}}{{\left| {{h_n}} \right|}^2}}}{{{\sigma ^2}}}}\!\! \right).$$
Here, ${p_{m,n}}$ is the transmit power for each pair of users, i.e. ${p_{m,n}} = {p_m} + {p_n}$. The coefficient $\frac{1}{2}$ is due to the fact that conventional OMA results in a multiplexing loss of $\frac{1}{2}$.

The problem presented in (\ref{equ:original_problem}) is a mixed-integer nonlinear programming (MINLP) problem, which is typically NP-hard and intractable to solve the global optimal solution. The key difficultly in solving (\ref{equ:original_problem}) arises from the binary scheduling variables, achievable data rate constraint, and objective function. To improve the tractability, in this paper, we decouple Problem (\ref{equ:original_problem}) into the subproblem of power allocation and user pairing, and solve the subproblems separately in an alternating manner.

\subsection{Power Allocation Optimization}
\label{sec:power_allocation_opt}
First, we optimize the transmit power of each user ${p_n}$ and ${p_m}$ for given user pairing ${x_{m,n}}$. The power allocation subproblem is given by
\begin{subequations}
\label{simfily_power_allocation}
\begin{align}
 \label{equ:power_allocation_R} \mathop {\mathrm{max} }\limits_{{p_n},{p_m}}\hspace{4.5mm}& \sum\limits_{m = 1}^{2K} {\sum\limits_{n \in {S_m}} {R^s_{n}} } \hfill \\
 \text{s.t.}\label{equ:power_allocation_Rn}\hspace{+5mm}&{R_{m,m}} \geqslant {R_m}, \hfill \\
  &\label{equ:power_allocation_Rm}{R_{n,n}} \geqslant {R_n} ,\hfill \\
 & \sum\nolimits_{m = 1}^{2K} {\sum\nolimits_{n \in {S_m}} {\left( {{p_n} + {p_m}} \right)} }  \leqslant P,\hfill
\end{align}
\end{subequations}
where ${S_m} = \left\{ {n\left| {{x_{m,n}} = 1} \right.} \right\}$. Despite the non-convexity of the subproblem, we can derive its closed-form solution, as follows.

According to (\ref{equ:power_allocation_Rn}), we have
\begin{align}
{p_n} \leqslant \frac{{{\sigma ^2}}}{{{{\left| {{h_m}} \right|}^2}}}\left(\sqrt {1 + \frac{{{p_{m,n}}{{\left| {{h_m}} \right|}^2}}}{{{\sigma ^2}}}}  - 1\right).
\end{align}
Similarly, according to (\ref{equ:power_allocation_Rm}), we have
\begin{align}
{p_n} \geqslant \frac{{{\sigma ^2}}}{{{{\left| {{h_n}} \right|}^2}}}\left(\sqrt {1 + \frac{{{p_{m,n}}{{\left| {{h_n}} \right|}^2}}}{{{\sigma ^2}}}}  - 1\right)
\end{align}

The first-order partial derivative of ${R^s_{n}}$ with respect to (w.r.t.) ${p_n}$ is
\begin{align}
&\label{first_order}
\frac{{\partial {R^s_{n}} }}{{\partial {p_n}}} = \frac{1}{{\ln 2}}\frac{{\left( {{{\left| {{h_n}} \right|}^2}{\sigma ^2} - {{\left| {{h_m}} \right|}^2}{\sigma ^2}} \right)}}{{({p_n}{{\left| {{h_n}} \right|}^2} + {\sigma ^2})({p_n}{{\left| {{h_m}} \right|}^2} + {\sigma ^2})}},
\end{align}
which is always non-negative, given ${\left| {{h_m}} \right|^2} < {\left| {{h_n}} \right|^2}$. As a result, ${R^s_{n}}$ is an increasing function of ${{p_n}}$,  and the optimal value of ${{p_n}}$, denoted by $p_n^*$, is given by
\begin{align}
\label{equ:power_allocation_pn_result}
{p^*_n} = \frac{{{\sigma ^2}}}{{{{\left| {{h_m}} \right|}^2}}}\left(\sqrt {1 + \frac{{{p_{m,n}}{{\left| {{h_m}} \right|}^2}}}{{{\sigma ^2}}}}  - 1\right).
\end{align}
By substituting $p_n^*$ into the objective function (\ref{equ:power_allocation_R}), we obtain
\begin{align}
	\label{equ:Rsn_equation}
{R^s_{n}}=& {\log _2}\left(1 + \frac{{{{\left| {{h_n}} \right|}^2}}}{{{{\left| {{h_m}} \right|}^2}}}\left(\sqrt {1 + \frac{{{p_{m,n}}{{\left| {{h_m}} \right|}^2}}}{{{\sigma ^2}}}}  - 1\right)\right) \nonumber \\
&- \frac{1}{2}{\log _2}\left(1 + \frac{{{p_{m,n}}{{\left| {{h_m}} \right|}^2}}}{{{\sigma ^2}}}\right).
\end{align}

Then, Problem (\ref{simfily_power_allocation}) can be equivalently rewritten as
\begin{subequations}
\label{simfily_power_allocation_pmn}
\begin{align}
	\label{equ:objective_further_derivation}
  \mathop {\mathrm{max} }\limits_{{p_{m,n}}}\hspace{2mm}& \sum\limits_{m = 1}^{2K} {\sum\limits_{n \in {S_n}} {R^s_{n}} }  \hfill \\
  	\label{equ:objective_further_derivation_constrain}
 \text{s.t.}\hspace{2.6mm}&\sum\nolimits_{m = 1}^{2K} {\sum\nolimits_{n \in {S_n}} {{x_{m,n}}{p_{m,n}}} }  \leqslant P.\hfill
\end{align}
\end{subequations}
Taking the second-order derivative of ${R^s_{n}}$ w.r.t. ${p_n}$ yields
\begin{equation}
	\frac{{{{d}^2}{R^s_{n}}}}{{{d}\,p_{m,n}^2}} =  - \frac{{{{\left| {{h_m}} \right|}_2}{{\left| {{h_n}} \right|}_2}\ln 2}}{{4{\sigma ^4}{{\left( {1 + \frac{{{{\left| {{h_m}} \right|}_2}{p_{m,n}}}}{{{\sigma ^2}}}} \right)}^{\frac{3}{2}}}}} < 0.
\end{equation}
Therefore, we can drawn the conclusion that (\ref{equ:objective_further_derivation}) is concave in ${p_{m,n}}$. In turn, Problem (\ref{simfily_power_allocation_pmn}) exhibits convexity and can be efficiently solved taking the Lagrange multiplier method. The Lagrange function of Problem (\ref{simfily_power_allocation_pmn}) is
\begin{align}
\label{Lagrange_equation}
  L({p_{m,n}},\!\upsilon )\!\!=& \!\!- \!\!\sum\limits_{m = 1}^{2K} {\sum\limits_{n \in {S_n}} {{{\!\!\!\log }_2}\bigg(\!\!1\!\! +\!\! \frac{{{{\left| {{h_n}} \right|}^2}}}{{{{\left| {{h_m}} \right|}^2}}}\Big(\sqrt {1\!\! +\!\! \frac{{{p_{m,n}}{{\left| {{h_m}} \right|}^2}}}{{{\sigma ^2}}}}  \!\!- \!\!1\Big)\!\!\bigg)} }  \hfill \nonumber\\
 &  \!+\!\! \sum\nolimits_{m = 1}^{2K} {\sum\nolimits_{n \in {S_n}} {\frac{1}{2}{{\log }_2}\Big(1 + \frac{{{p_{m,n}}{{\left| {{h_m}} \right|}^2}}}{{{\sigma ^2}}}\Big)} }  \hfill \nonumber\\
  & \!+\!\! \upsilon \sum\nolimits_{m = 1}^{2K} {\sum\nolimits_{n \in {S_n}} ({{p_{m,n}}} }  - P) ,\hfill
\end{align}
where $\upsilon  > 0$ is the dual variable corresponding to  \eqref{equ:objective_further_derivation_constrain}.

After taking the first-order partial derivative of (\ref{Lagrange_equation}) w.r.t. ${p_{m,n}}$, the KKT conditions of (\ref{simfily_power_allocation_pmn}) are given by
\begin{align}
\frac{{\partial L({p_{m,n}},\upsilon )}}{{\partial {p_{m,n}}}} = {\alpha ^3} \hspace{-0.5mm}- \hspace{-0.5mm}\frac{{{{\left| {{h_n}} \right|}^2} \hspace{-0.5mm}-\hspace{-0.5mm} {{\left| {{h_m}} \right|}^2}}}{{{{\left| {{h_n}} \right|}^2}}}{\alpha ^2}\hspace{-0.5mm} - \hspace{-0.5mm} \frac{{{{\left| {{h_m}} \right|}^2}({{\left| {{h_n}} \right|}^2} \hspace{-0.5mm}- \hspace{-0.5mm}{{\left| {{h_m}} \right|}^2})}}{{2\ln 2{\sigma ^2}\upsilon {{\left| {{h_n}} \right|}^2}}},\nonumber
\end{align}
where $\alpha  = \sqrt {\frac{{{p_{m,n}}{{\left| {{h_m}} \right|}^2}}}{{{\sigma ^2}}} + 1} >1$.

Since the optimal solution to Problem \eqref{simfily_power_allocation_pmn} satisfies $\frac{{\partial L({p_{m,n}},\upsilon )}}{{\partial {p_{m,n}}}} = 0 $, we define $f(\alpha): \mathbb{R}^+\to \mathbb{R}$ as
\begin{align}
\label{equation_am}
f(\alpha ) = {\alpha ^3} - \frac{{{{\left| {{h_n}} \right|}^2} - {{\left| {{h_m}} \right|}^2}}}{{{{\left| {{h_n}} \right|}^2}}}{\alpha ^2} - \frac{{{{\left| {{h_m}} \right|}^2}({{\left| {{h_n}} \right|}^2} - {{\left| {{h_m}} \right|}^2})}}{{2\ln 2{\sigma ^2}\upsilon {{\left| {{h_n}} \right|}^2}}}.\nonumber
\end{align}
%{\color{black}Please describe why you define $f(\alpha)$.}
Note that $f(\alpha) = 0$ holds at the optimal $p_{m,n}$.
Taking the first-order derivative of $f(\alpha )$ w.r.t. $\alpha $, we have
\begin{align}
\frac{{\mathrm{d} f(\alpha )}}{{\mathrm{d} \alpha }}
= 3\alpha\left[\alpha - \frac{2}{3}\left(1-\frac{{ {{\left| {{h_m}} \right|}^2}}}{{{{\left| {{h_n}} \right|}^2}}} \right)\right].
\end{align}
Then, setting $\frac{{\mathrm{d} f(\alpha )}}{{\mathrm{d} \alpha }} = 0$ yields two roots:
\begin{align}
{\alpha _1} = 0, \text{ }{\alpha _2}  \!=\! \frac{2}{3}\left(1 - \frac{{{{\left| {{h_m}} \right|}^2}}}{{{{\left| {{h_n}} \right|}^2}}}\right) < \frac{2}{3}.
\end{align}
%{\color{black}Please describe how you solve (18) here. Surprisingly, many reviewers do not know to solve quadratic equations.}

We can analyze the number of positive roots of $f(\alpha) = 0 $ based on the monotonicity of $f(\alpha)$.
Since $f(\alpha_1) < 0$, $f(\alpha_2) <0$, and $\alpha_1 < \alpha_2$, according to the monotonicity  of a cubic function, there is only one positive root of $f(\alpha)=0$:
\begin{equation}
	\label{solution_a}
	\alpha  = \frac{{\sqrt[3]{a_{m,n}}}}{{3 \cdot {\sqrt[3]{4}}}} + \frac{{{\sqrt[3]{4}}{{({{\left| {{h_n}} \right|}^2} - {{\left| {{h_m}} \right|}^2})}^2}}}{{3\sqrt[3]{a_{m,n}}{{\left| {{h_n}} \right|}^4}}} + \frac{{{{\left| {{h_n}} \right|}^2} - {{\left| {{h_m}} \right|}^2}}}{{3{{\left| {{h_n}} \right|}^2}}},
\end{equation}
where $a_{m,n}$ is given by
%\begin{figure*}[b]
 %   \hrulefill
    \begin{align}
        {a_{m,n}} &= 4 - \frac{{4{{\left| {{h_m}} \right|}^6}}}{{{{\left| {{h_n}} \right|}^6}}}\! +\! {\left| {{h_m}} \right|^2}\left(\frac{{27}}{{\ln 2{\sigma ^2}\upsilon }} \!-\! \frac{{12}}{{{{\left| {{h_n}} \right|}^2}}} \right) \notag\\
        +& \frac{{3{{\left| {{h_m}} \right|}^4}}}{{{{\left| {{h_n}} \right|}^4}}}\left(4 - \frac{{9{{\left| {{h_n}} \right|}^2}}}{{\ln 2{\sigma ^2}\upsilon }}\right) + 3\sqrt 3 \frac{{\left| {{h_m}} \right|\left({{\left| {{h_n}} \right|}^2} - {{\left| {{h_m}} \right|}^2}\right)}}{{{{\left| {{h_n}} \right|}^4}\ln 2{\sigma ^2}\upsilon }}  \notag\\
\times &\sqrt {8\ln 2{\sigma ^2}\upsilon \left({{\left| {{h_n}} \right|}^2} - {{\left| {{h_m}} \right|}^2}\right) + 27{{\left| {{h_n}} \right|}^4}{{\left| {{h_m}} \right|}^2}}.
    \end{align}
%\end{figure*}
We solve \eqref{solution_a} using Formula of Cardano \cite{Bronshtein2015} and choose the positive root. The Formula of Cardano is a widely-used approach for solving cubic equations.

When $x_{m,n}=1$, by substituting
$\alpha  = \sqrt {\frac{{{p_{m,n}}{{\left| {{h_m}} \right|}^2}}}{{{\sigma ^2}}} + 1} $ into (\ref{solution_a}), the optimal power allocation,
 denoted by $p_{m,n}^*$, can be obtained in closed-form, as given by
\begin{align}
	\label{equ:pmn_results}
p_{m,n}^{*}=&\dfrac{\sigma ^2}{\left| h_m \right|^2}\left[ \left( \frac{\sqrt[3]{a_{m,n}}}{\sqrt[3]{4}}+\frac{\sqrt[3]{4}\left( \left| h_n \right|^2-\left| h_m \right|^2 \right)}{3\sqrt[3]{a_{m,n}}\left| h_n \right|^4} \right. \right.
 \nonumber\\
&\hspace{30mm}
\left. \left. +\frac{\left| h_n \right|^2-\left| h_m \right|^2}{3\left| h_n \right|^2} \right) -1 \right] ,
\end{align}
Moreover, according to \eqref{equ:power_allocation_pn_result} and \eqref{equ:pmn_results}, the optimal $p_m$, denoted by $p_m^*$, can be obtained since $p_m^*=p_{m,n}^*-p_n^*$. By adjusting the dual variable $\upsilon$ until
% \begin{equation}
% 	\label{equ:solve_the_semi_closed_solution}
	$\sum_{k=1}^{2K}p_k^*=P, \, \forall k=1,\cdots, 2K$,
% \end{equation}
we obtain the optimal transmit power $p_k^*,\, \forall k$.
%{\color{black}Do we have any insight here based on (21)?}

\subsection{User Pairing Optimization}
Given the power allocation $p_n^*$ and $p_m^*$, we can relax the binary variables ${x_{m,n}} \!\!\in\!\! \left\{0, 1\right\}$ into continuous variables $\hat{x}_{m,n}\in [0,1]$. Problem \eqref{equ:original_problem} can be recast as
\begin{subequations}
    \label{equ:relaxation_problem}
	\begin{align}
    \label{equ:objective_pairing_relaxation}
    \mathop{\mathrm{max}}_{x_{m,n}} \hspace{+2.5mm}& \sum_{m=1}^{2K}\sum_{n=1}^{2K} {\hat{x}_{m,n}}  {R^s_{n}}\\
    \label{equ:rm_pairing_relaxation}
   \text{s.t.} \hspace{+3.5mm}& {R_{m,m}} \geqslant {\hat{x}_{m,n}} R_m,\\
    \label{equ:rn_pairing_relaxation}
    &  {R_{n,n}} \geqslant {\hat{x}_{m,n}} R_n,\\
    \label{equ:sum_power_pairing_relaxation}
    &  \sum\nolimits_{m = 1}^{2K} {\sum\nolimits_{n = m + 1}^{2K} {{\hat{x}_{m,n}}\left( {{p_n} + {p_m}} \right)} }  \leqslant P,\\
    \label{equ:integer_pairing_relaxation}
    &  0 \leqslant {\hat{x}_{m,n}}, \leqslant 1,1 \leqslant m,n \leqslant 2K,\\
    \label{equ:xmm_0_pairing_relaxation}
    &   {\hat{x}_{m,n}} = 0, \\
    \label{equ:sum_m_pairing_relaxation}
    &   \sum\nolimits_{m = 1}^{2K} {\hat{x}_{m,n}}  = 1,1 \leqslant n \leqslant 2K,  \\
    \label{equ:last_constraints}
& \sum\nolimits_{n = 1}^{2K} {\hat{x}_{m,n}}  = 1,1 \leqslant m \leqslant 2K,
\end{align}
\end{subequations}
Here, $\hat{x}_{m,n}  \in [0,1]$ can be interpreted as how likely users $m$ and $n$ are assigned to form a NOMA group and share the same resource block.
% \sout{Furthermore, $\mathbf{X}\!=\!\left\{x_{m,n}\right\}\in \mathbb{R}^{2K\times 2K}$ is actually a matrix of fuzzy assignment.}

% {\color{black}Whhat is fuzzy assignment?}

%Let ${\color{black}\mathbf{a}_1} \preccurlyeq {\color{black}\mathbf{a}_2}$ define componentwise inequality between two vectors ${\color{black}\mathbf{a}_1} $ and ${\color{black}\mathbf{a}_2}$; i.e., each element of ${\color{black}\mathbf{a}_1}$ is no greater than the corresponding element of ${\color{black}\mathbf{a}_2}$.
% {\color{black}What does this mean? Matrix inequality? You may just say this ``indicates that any element of $\mathbf{A}$ is no larger than the corresponding element of $\mathbf{B}$.''}
By vectorization, Problem (\ref{equ:relaxation_problem}) is rewritten as
\begin{subequations}
    \label{equ:relaxation_vector}
\begin{align}
    \label{equ:objective_pairing_vec}
    \mathop{\mathrm{max}}_{\color{black}\hat{\mathbf{x}}} &\quad \mathbf{r}_s^T{\color{black}\hat{\mathbf{x}}}\\
    \label{equ:rm_pairing_vec}
    \text{s.t.} &\hspace{3mm}{\color{black}\hat{\mathbf{x}}} \preccurlyeq \mathbf{b},\\
    \label{equ:sum_power_pairing_vec}
    & \hspace{3mm} \mathbf{p}^T{\color{black}\hat{\mathbf{x}}} \leqslant P,\\
    \label{equ:integer_pairing_vec}
    & \hspace{3mm} -{\color{black}\hat{\mathbf{x}}} \preccurlyeq  \mathbf{0},\\
    \label{equ:sum_m_pairing_rvec}
    & \hspace{3mm} \mathbf{D}{\color{black}\hat{\mathbf{x}}} = \mathbf{1},
\end{align}
\end{subequations}
where ${\hat{\mathbf{x}}\!\in\! \mathbb{R}^{K(2K-1)}}$ and $\mathbf{p}\!\!\in\!\! \mathbb{R}^{K(2K-1)}$ are the vectorization of the elements above the main diagonal of ${\color{black}\hat{\mathbf{X}}}$  and $\left\{p_m+p_n\right\}$ in the row-major order, respectively.
Then, the $\left[\frac{1}{2}\left(4K-m\right)\left(m-1\right)+n-m\right]$-th elements of ${\color{black}\hat{\mathbf{x}}}$, $\mathbf{p}$, $\mathbf{r}_s\!\in\! \mathbb{R}^{K(2K-1)}$, and $\mathbf{b}\in \mathbb{R}^{K(2K-1)}$ are ${\color{black} \hat{x}_{m,n}}$, $p_m+p_n$, $R^s_{n}$ and $\min\left\{\frac{{R_{m,m}}}{R_m}, \frac{{R_{n,n}}}{R_n}, 1\right\}$, respectively. Moreover, the $n$-th row of $\mathbf{D}\in \mathbb{R}^{2K\times K(2K-1)}$, denoted by $\mathbf{d}_n^T$, satisfies
${\color{black}\mathbf{d}_n^T\hat{\mathbf{x} }= \sum_{m-1}^{2K}\hat{x}_{m, n}}$, according to constraints  \eqref{equ:sum_m_pairing_relaxation} and \eqref{equ:last_constraints}.

By combining \eqref{equ:rm_pairing_vec}, \eqref{equ:sum_power_pairing_vec}, and \eqref{equ:integer_pairing_vec}, Problem \eqref{equ:relaxation_vector} can be further rewritten as
\begin{subequations}
    \label{equ:newton_pairing_matrix}
\begin{align}
    \label{equ:objective_pairing_matrix}
    \mathop{\mathrm{max}}_{\hat{\mathbf{x}}} &\hspace{3.8mm} \mathbf{r}_s^T\hat{\mathbf{x}}\\
    \label{equ:rm_pairing_matrix}
    \text{s.t. } & \hspace{3mm}\mathbf{A} \hat{\mathbf{x}} \preccurlyeq \mathbf{u},\\
    \label{equ:sum_m_pairing_matrix}
    & \hspace{3mm}\mathbf{D}\hat{\mathbf{x}} = \mathbf{1},
\end{align}
\end{subequations}
where $\mathbf{A}=\left[\mathbf{I}, -\mathbf{I}, \mathbf{p}\right]^T\in\mathbb{R}^{\left[2K(2K-1)+1\right]\times K(2K-1)}$ and $\mathbf{u}=\left[\mathbf{b}^T, \mathbf{0}^T, P\right]^T\in \mathbb{R}^{2K(2K-1)+1}$.

The logarithmic barrier and Simplex methods~\cite{Ficken2015} are widely adopted by various Linear Programming (LP) solvers, e.g.,
%Gurobi~\cite{gurobi2022},
%CPLEX~\cite{cplex2021},
CVX Toolbox~\cite{LiTMC2016,LyuJSAC2017,Lyu2018}.
%~\cite{cvx2014},
%and GLPK~\cite{glp2010}.
Since $\mathbf{D}$ is sparse, the barrier method is more effective than the Simplex method in solving sparse LP problems \cite{LyuTCOM2018}.
%\cite{PingQi2014}.

We utilize the logarithmic barrier method~\cite{Boyd2004} to solve Problem (\ref{equ:newton_pairing_matrix}).
%{\color{black}Please describe the logarithmic barrier method here, and also why you choose it. You need to show that the algorithm is suitable for the purpose here.}
In the method, the optimization problem is modified by adding a logarithmic barrier function to the objective function. The barrier function penalizes the constraints, encouraging the optimization to move towards feasible solutions~\cite{Boyd2004}.
The barrier function is typically the sum of the negative logarithms of the variables that define the feasible region. The barrier function we choose is
\begin{align}
    \phi\left(\hat{\mathbf{x}}\right) = - \sum_{i=1}^{2K(2K-1)+1} \ln \left(u_i-\mathbf{a}_i^T\hat{\mathbf{x}}\right),
\end{align}
where $u_i\in \mathbb{R}$ and $\mathbf{a}_i\in \mathbb{R}^{K(2K-1)}$ are the $i$-th rows of $\mathbf{u}$ and $\mathbf{A}$, respectively.

% {\color{black} The logarithmic barrier method comprises a damped Newton phase and a quadratically convergent phase. During the damped Newton phase, the method attempts to find a solution to the problem considered using a variant of Newton's method.
% %which is an iterative optimization technique that uses the Hessian matrix of the objective function to find a local minimum.
% A damping factor ensures that the Hessian matrix is always positive definite, which helps to ensure convergence.}
%{\color{black}In the damped Newton phase,
Let $t >0$ denote the parameter (or step size) of the logarithmic barrier method.
%, which regulates the approximation accuracy of \eqref{equ:newton_pairing_matrix}.
Problem (\ref{equ:newton_pairing_matrix}) is then rewritten as
\begin{subequations}
    \label{equ:barrier_problem}
\begin{align}
    \label{equ:objective_pairing_barrier}
    \mathop{\mathrm{min}}_{\mathbf{x}} & \hspace{3mm}g\left(\hat{\mathbf{x}}\right) = -t\mathbf{r}_s^T\hat{\mathbf{x}} + \phi\left(\hat{\mathbf{x}}\right)\\
    \label{equ:sum_m_pairing_barrier}
    &\hspace{-4mm} \text{s.t.}\hspace{3mm}  \mathbf{D}\hat{\mathbf{x}} = \mathbf{1}.%\label{eq. equality cosntraint}
\end{align}
\end{subequations}
Let ${y_i=\frac{1}{u_i-\mathbf{a}_i^T\hat{\mathbf{x}}}}$ be the $i$-th element of $\mathbf{y}{\in \mathbb{R}^{2K(2K-1)+1}}$  and $\mathbf{w}\in \mathbb{R}^{2K}$ be the Lagrange multiplier associated with (\ref{equ:sum_m_pairing_barrier}).

%{\color{black} Given fixed $t$, Problem \eqref{equ:barrier_problem} can be solved using the infeasible start Newton method~\cite{Zhang2020a, Weihs2013}.}
%For a given $t$,

In each iteration of the logarithmic barrier method, we update $t$ by  $t:=\xi t$, which regulates the accuracy of using \eqref{equ:barrier_problem} to approximate \eqref{equ:newton_pairing_matrix}. Here, $\xi>1$ is a preconfigured coefficient. Given the fixed $t$, the infeasible start Newton method~\cite{Zhang2020a, Weihs2013} is adopted to solve \eqref{equ:barrier_problem} iteratively.

The infeasible start Newton method starts by evaluating the primal and dual Newton steps $\Delta \mathbf{w}\in \mathbb{R}^{2K}$ and ${\Delta \hat{\mathbf{x}}\in \mathbb{R}^{K(2K-1)}}$.
Given $t$, the primal and dual Newton steps of Problem \eqref{equ:barrier_problem} are given by

	\begin{equation}
	\label{equ:primal_newton_step}
	\mathbf{K}\left(\hat{\mathbf{x}}, \mathbf{w}\right)\cdot
	\begin{bmatrix}
		\Delta \hat{\mathbf{x}}\\
		\mathbf{w} + \Delta \mathbf{w}
	\end{bmatrix}
	=
	-
	\begin{bmatrix}
		-t\mathbf{r}_s+\mathbf{A}\mathbf{y}\\
		\mathbf{D\hat{x}}-\mathbf{1}
	\end{bmatrix},
\end{equation}
where $\mathbf{K}\left(\mathbf{\hat{x}}, \mathbf{w}\right)$ is the Karush-Kuhn-Tucker (KKT) matrix~\cite{vanderbei2020} and is given by
\begin{equation}\label{eq: coefficient matrix}
	\mathbf{K} \left(\hat{\mathbf{x}}, \mathbf{w}\right)= \begin{bmatrix}
		\mathbf{A}^T\mathrm{diag}\left(\mathbf{y}\right)\mathbf{A} & \mathbf{D}^T \\
		\mathbf{D} & \mathbf{0}
	\end{bmatrix}
\end{equation}		
We utilize the LU decomposition~\cite{Golub2013} to solve \eqref{equ:primal_newton_step} for $\Delta \hat{\mathbf{x}}$ and $\Delta \mathbf{w}$,
so that we can avoid computationally expensive matrix inversions~\cite{Heath2018}. Let $\mathbf{L},\mathbf{U} \in \mathbb{R}^{K(2K+1)\times K(2K+1)}$ denote the lower and higher triangular matrices, respectively, and
\begin{equation}
		\label{equ:LU_Decomposition}
	\mathbf{LU} = \mathbf{K}\left(\hat{\mathbf{x}}, \mathbf{w}\right).
\end{equation}
By using the forward and back substitution algorithms \cite{Cassel2021}, we can derive $\Delta \hat{\mathbf{x}}$ and $\Delta \mathbf{w}$.

Define $\mathbf{J}: \mathbb{R}^{K(2K-1)}\times \mathbb{R}^{2K}\to \mathbb{R}^{K(2K-1)}\times \mathbb{R}^{2K}$ as
	\begin{equation}
		\mathbf{J}\left(\mathbf{\hat{x}},  \mathbf{w}\right)	= \left(\nabla g\left(\hat{\mathbf{x}}\right) + \mathbf{D}^T\mathbf{w}, \mathbf{D\hat{x}}-\mathbf{1}\right).
	\end{equation}
The Frobenius norm of $\mathbf{J}\left(\mathbf{\hat{x}}, \mathbf{w}\right)$ is given by
	\begin{equation}
		\begin{Vmatrix}
			\mathbf{J}\left(\hat{\mathbf{x}}, \mathbf{w}\right)
		\end{Vmatrix}_F
		=	\sqrt{
			\begin{Vmatrix}
				\nabla g\left(\hat{\mathbf{x}}\right) + \mathbf{D}^T\mathbf{w}	
			\end{Vmatrix}_F^2
			+
			\begin{Vmatrix}
				\mathbf{D}\hat{\mathbf{x}}-\mathbf{1}
			\end{Vmatrix}_F^2
		}\,.
\end{equation}

Next, we utilize backtracking line search to produce the step size by $s:=\tau s$ for updating $\mathbf{\hat{x}}$ and $\mathbf{w}$, i.e.,
\begin{equation}\label{eq: update x and w}
	\mathbf{\hat{x}} := \mathbf{\hat{x}} + s\Delta \mathbf{\hat{x}}\quad \text{ and }\quad
	\mathbf{w}:=\mathbf{w} + s\Delta\mathbf{w}
\end{equation}
\begin{equation}\label{eq: convergence criterion}
{\rm until }
\begin{Vmatrix} \mathbf{J}\!\left(\hat{\mathbf{x}}\!\! +\!\! {\color{black} s} \Delta \hat{\mathbf{x}}, \mathbf{w}\!\!+\!\!{\color{black} s}\Delta\mathbf{w}\right)\end{Vmatrix}_F
\!\!{\color{black}\leqslant }\!\!
\left(1\!-\!\zeta {\color{black} s}\right)\begin{Vmatrix} \mathbf{J}\!\left(\mathbf{\hat{x}}, \mathbf{w}\right) \end{Vmatrix}_F.
\end{equation}
Here, $\tau\in (0,1)$ and $\zeta \in (0,\frac{1}{2})$ are preconfigured coefficients.

Upon the stopping criterion \eqref{eq: convergence criterion} is satisfied,
the updated $\mathbf{\hat{x}}$ and $\mathbf{w}$ are substituted into \eqref{equ:primal_newton_step} and \eqref{eq: coefficient matrix} to update $\Delta \hat{\mathbf{x}}$ and $\Delta \mathbf{w}$, followed by the updating of $\mathbf{\hat{x}}$ and $\mathbf{w}$ using \eqref{eq: update x and w}. This repeats until $\begin{Vmatrix} \mathbf{J}\left(\mathbf{\hat{x}}, \mathbf{w}\right) \end{Vmatrix}_F$ is smaller than a predefined, sufficiently small threshold, e.g., $\rho$, and \eqref{equ:sum_m_pairing_barrier} is satisfied.

Let $L$ and $S$ denote constants satisfying:
$\forall \left(\hat{\mathbf{x}}_i, \mathbf{w}_i\right)$, $i=1, 2$,
\begin{equation}
	\begin{Vmatrix}
		\mathbf{K}\!\left(\!\mathbf{\hat{x}}_1, \!\mathbf{w}_1\!\right)
		\!\!-\!\!
		\mathbf{K}\!\left(\!\mathbf{\hat{x}}_2, \!\mathbf{w}_2\!\right)
	\end{Vmatrix}_F
	\!\!\leqslant\!\!
	L\!
	\sqrt{
		\begin{Vmatrix}
			\mathbf{\hat{x}}_1\!\!-\!\!\mathbf{\hat{x}}_2
		\end{Vmatrix}_2^2
		\!+\!
		\begin{Vmatrix}
			\mathbf{w}_1\!\!-\!\!\mathbf{w}_2
		\end{Vmatrix}_2^2
	};
\end{equation}
%and
\begin{equation}S \geqslant \begin{Vmatrix}
	\mathbf{K}\left(\hat{\mathbf{x}}, \mathbf{w}\right)^{\dag}
\end{Vmatrix}_F.
\end{equation}
The step size $s$, obtained by backtracking line search, satisfies $s<1$ in the damped Newton phase if $\begin{Vmatrix}\mathbf{J}\left(\mathbf{\hat{x}}, \mathbf{w}\right)\end{Vmatrix}_F>\frac{1}{S^2L}$~\cite{Boyd2004}. Hence, $\begin{Vmatrix}\mathbf{J}\left(\mathbf{\hat{x}}, \mathbf{w}\right)\end{Vmatrix}$ is reduced in each iteration \cite{Andrei2022}.
Once the damped Newton phase has reasonably converged, i.e., $\begin{Vmatrix}\mathbf{J}\left(\mathbf{\hat{x}},  \mathbf{w}\right)\end{Vmatrix}_F\leqslant \frac{1}{S^2L}$, the logarithmic barrier method enters the quadratically convergent phase, where the step size is $s=1$ and the error converges quadratically to zero~\cite{Feller2017}. This allows the algorithm to find a high-precision solution in relatively few iterations.

When the infeasible start Newton method converges, $t$ is updated by $t:=\xi t$ and then the infeasible start Newton method restarts. This repeats until $\frac{1}{t}K\left(2K-1\right)<\epsilon$, where $\epsilon$ indicates the approximation accuracy of \eqref{equ:barrier_problem} with regards to \eqref{equ:newton_pairing_matrix}. The output of the logarithmic barrier method is the continuous relaxation of user pairing, i.e., $\hat{\mathbf{X}}$.

%{\color{black}\textbf{COMMENT:} The way how (27) was solved is very unclear. Please try to describe how it is done by talking about the damped Newton phase and the quadratic convergence phase here. Moreover, $\mathbf{J}\left(\mathbf{x}, \mathbf{w}\right)$ needs to be defined here.}

Finally, given the continuous $\hat{\mathbf{X}}$, we utilize a greedy method that iteratively chooses the most probable pairs. $\mathbf{U}$ is initialized to be empty, i.e., $\mathbf{U}=\varnothing$ initially.
In each iteration, we choose and record the pair $\left\{m, n\right\}$ with the largest $\hat{x}_{m,n}$ from unrecorded pairs, i.e., $\mathbf{U}\cap \left\{m,n\right\} = \varnothing$, since users $m$ and $n$ have the highest pairing probability among all users not recorded in $\mathbf{U}$ yet.

The algorithm of user pairing is summarized in Alg. \ref{alg:user_pairing}, where ${\color{black} \hat{\mathbf{x}}^*}$ denotes the optimum of (\ref{equ:newton_pairing_matrix}).

\begin{algorithm}
    \caption{User Pairing}\label{alg:user_pairing}
    \SetKw{Continue}{continue}
    \KwData{Initialize $t=t^{(0)}$, $\xi > 1$, $\epsilon > 0$, $\rho > 0$, control factors in backtracking line search $\zeta\in \left(0, 0.5\right)$, $\tau\in \left(0, 1\right)$, pairing set $\mathbf{U}=\varnothing $, $\mathbf{X}=\mathbf{0}$}
    \tcc{logarithmic barrier-based approach to obtain the assignment}
    \While{ $\frac{1}{t}K\left(2K-1\right)<\epsilon$}
    {
        %\sout{Using infeasible start Newton method to compute $\mathbf{x}^*$ by solving (\ref{equ:barrier_problem})}\\

         	\tcc{Using infeasible start Newton method to compute \eqref{equ:barrier_problem} with the given~$t$}
         	
         	{$s:=1$}
         	
         \While{$\mathbf{D\hat{x}}=\mathbf{1}$ \&\& $\begin{Vmatrix}\mathbf{J}\left(\hat{\mathbf{x}}, \mathbf{w}\right) \end{Vmatrix}_F\leqslant \rho$}
         {
         	Calculate $\Delta \hat{\mathbf{x}}$ and $\Delta \mathbf{w}$ in (\ref{equ:primal_newton_step})\\
         	\tcc{Backtracking line search to obtain step size $s$}
         	\While{$\begin{Vmatrix} \mathbf{J}\left(\hat{\mathbf{x}} + s\Delta \hat{\mathbf{x}}, \mathbf{w}+s\Delta\mathbf{w}\right)\end{Vmatrix}_F>\left(1-\zeta s\right)\begin{Vmatrix} \mathbf{J}\left(\mathbf{\hat{x}}, \mathbf{w}\right) \end{Vmatrix}_F$}
         	{$s:=\tau s$}
%         	Backtracking line search to obtain $t$\\
         	$\hat{\mathbf{x}}^*:=\hat{\mathbf{x}} + s\Delta \hat{\mathbf{x}}$ and  $\mathbf{w}:=\mathbf{w} + s\Delta \mathbf{w}$
         }
        }
        \tcp{update $\mathbf{\hat{x}}$ and $t$}
        $\hat{\mathbf{x}}:=\hat{\mathbf{x}}^*$\\
        $t:=\xi t$

      obtain the assignment $\hat{\mathbf{x}}$.\\
     \tcc{greedy-based approach to obtain the user pairing strategy}
    \For{$m=1,\cdots, 2K-1$}
    {
        \For{$n=m+1, \cdots, 2K$}
        {
            \If{$\mathbf{U}\cap \left\{m,n\right\} = \varnothing $}
            {
                Choose the largest element $\hat{x}_{m,n}$ and set $x_{m,n}=1$\\
                $\mathbf{U} = \mathbf{U} \cup \left\{m, n\right\}$\\
                Set {$\hat{x}_{m,n_0} = \hat{x}_{m_0, n} = -\infty$} for $n_0=1,\cdots, 2K$, and $m_0=1, \cdots, 2K$
                }
            \Else
            {
                \Continue
            }
        }
    }
    \Return{the user pairing strategy $\mathbf{X}=\left\{x_{m,n}\right\}$}
\end{algorithm}

%{\color{black} \textbf{COMMENT:} How about the overall algorithm? Does user pairing and power allocation need to be solved in an alternating manner, i.e., by using alternating optimization? If this is the case, we shall say why alternating optimization converges here, which is straightforward though.}

\subsection{Algorithm Summary}
The overall algorithm is illustrated in Alg. \ref{alg:overall}, which consists of two phases (i.e., power allocation and user pairing) operating in an alternating manner. In the power allocation phase, we utilize \eqref{equ:power_allocation_pn_result} and \eqref{equ:pmn_results}
%and \eqref{equ:solve_the_semi_closed_solution}
to obtain the transmit power of each user with a given fixed user pairing strategy $\mathbf{X}$. In the user pairing phase, Alg. \ref{alg:user_pairing} is executed to produce the user pairing strategy given the fixed transmit powers of all users.
The user pairing strategy is then input to the power allocation to start the next iteration of the power allocation and user pairing phases.
Let $o_q$ and $\eta$ denote the sum secrecy rate in the $q$-th iteration of Alg. \ref{alg:overall} and the tolerance, respectively. If $\left|o_q-o_{q-1}\right|<\eta$, Alg. \ref{alg:overall} returns the user pairing strategy $\mathbf{X}$ and transmit power of each user $p_n$, $n=1, \cdots, 2K$.

\begin{algorithm}
	\caption{Overall Algorithm}\label{alg:overall}
	\SetKw{Continue}{continue}
	\KwData{Initial $o_q = +\infty$, $o_{q-1}=-\infty$, counter $q=0$, tolerance $\eta>0$}
	% \tcc{logarithmic barrier-based approach to obtain the assignment}
	\While{$\left|o_q - o_{q-1}\right| < \eta$}
	{
		\tcp{record the previous sum secrecy rate}
		$o_q=\sum_{m=1}^{2K}\sum_{n=m+1}^{2K}x_{m,n}R^s_{n}$\\
		\tcp{power allocation phase}
		Use \eqref{equ:power_allocation_pn_result} and \eqref{equ:pmn_results}
		%and \eqref{equ:solve_the_semi_closed_solution}
		to obtain the optimal power allocation $p_{m,n}^*$\\
		\tcp{user pairing phase}
		Use Alg. \ref{alg:user_pairing} to obtain the pairing strategy $\mathbf{X}=\left\{x_{m,n}\right\}$\\
		\tcp{update the counter}
		$q:=q+1$\\
		\tcp{record the sum secrecy rate}
		$o_q =\sum_{m=1}^{2K}\sum_{n=m+1}^{2K}x_{m,n}R^s_{n}$
%			\tcc{Using infeasible start Newton method to compute (\ref{equ:barrier_problem})}
%			\While{$\mathbf{Dx}=\mathbf{1}$ and $\begin{Vmatrix}\mathbf{J}\left(\mathbf{x}, \mathbf{w}\right) \end{Vmatrix}_F\leqslant \rho$}
%			{
%				Calculate $\Delta \mathbf{x}$ and $\Delta \mathbf{w}$ in (\ref{equ:primal_newton_step})\\
%				\tcc{Backtracking line search to obtain step size $s$}
%				\While{$\begin{Vmatrix} \mathbf{J}\left(\mathbf{x} + t\Delta \mathbf{x}, \mathbf{w}+t\Delta\mathbf{w}\right)\end{Vmatrix}_F>\left(1-\zeta t\right)\begin{Vmatrix} \mathbf{J}\left(\mathbf{w}, \mathbf{w}\right) \end{Vmatrix}_F$}
%				{$s:=\tau s$}
%				%         	Backtracking line search to obtain $t$\\
%				$\mathbf{x}^*:=\mathbf{x} + s\Delta \mathbf{x}$, $\mathbf{w}:=\mathbf{w} + s\Delta \mathbf{w}$
%			}
%		$\mathbf{x}:=\mathbf{x}^*$\\
%		$t:=\xi t$
	}
%	obtain {\color{black} the} assignment $\mathbf{x}$.\\
%	\tcc{greedy-based approach to obtain the user pairing strategy}
%	\For{$m=1,\ldots, 2K-1$}
%	{
%		\For{$n=m+1, \ldots, 2K$}
%		{
%			\If{$\mathbf{U}\cap \left\{m,n\right\} = \varnothing $}
%			{
%				choose largest $x_{m,n}$\\
%				$\mathbf{U} = \mathbf{U} \cup \left\{m, n\right\}$\\
%				set $x_{m,n_0} = x_{m_0, n} = -\infty$ for $n_0=1,\ldots 2K$, $m_0=1, \ldots, 2K$
%			}
%			\Else
%			{
%				\Continue
%			}
%		}
%	}
	\Return{the user pairing strategy $\mathbf{X}=\left\{x_{m,n}\right\}$ and the power allocation $\left\{p_{n}\right\}$, $m, n=1, \ldots, 2K$}
\end{algorithm}

\subsection{Convergence Analysis}
\subsubsection{Convergence of User Pairing}
\label{sec:convergence_user_pairing}
We analyze the convergence rate of the LP relaxation in Alg.~\ref{alg:user_pairing}.
According to~\cite[eq. (11.13)]{Boyd2004}, the LP in Alg. \ref{alg:user_pairing} requires $N_{\mathrm{LP}}$ iterations to adjust the parameter $t$ and guarantee the desired accuracy level of $\epsilon$: %, where $N_{\mathrm{LP}}$ is given by
\begin{equation}
N_{\mathrm{LP}} = \left\lceil \log \left( \frac{K\left( 2K-1\right)  }{\epsilon t^{\left( 0\right)  }} \right)  /\log \left( \xi \right)   \right\rceil,
	\end{equation}
where $t^{(0)}$ is the initial value of $t$.

In each of the LP iterations, backtracking line search is carried out to search for the step size $s$.
According to \cite{Boyd2004}, the backtracking line search in the damped Newton phase uses fewer than $N_l = \left\lceil \log\left(S^2L\kappa 	\right)/\log\left(\frac{1}{\tau}\right) \right\rceil$ iterations to choose the step size $s$. Here, ${\color{black}\kappa = \begin{Vmatrix}\mathbf{J}\left(\hat{\mathbf{x}}^{(0)}, \mathbf{w}^{(0)}\right) \end{Vmatrix}_F}$, where ${\color{black}\hat{\mathbf{x}}^{(0)}}$ and $\mathbf{w}^{(0)}$ are the initial ${\color{black}\hat{\mathbf{x}}}$ and $\mathbf{w}$, respectively.
According to~\cite{Gallier2020}, the damped Newton phase takes $	N_\mathrm{D} = \left\lceil S^2L\kappa/\zeta\tau	\right\rceil$ iterations to achieve $\begin{Vmatrix}\mathbf{J}\left(\mathbf{\hat{x}}, \mathbf{w}\right)\end{Vmatrix}_F\leqslant \frac{1}{S^2L}$ before the commencement of the quadratically convergent phase.
%{\color{black}\textbf{ COMMENT: }How about $s$?} {\color{brown}\textit{COMMENT: Actually, there is no $s$ in $N_\mathrm{D}$, since the step size $s$ is under the control of $\zeta$ and $\tau$}}
In the quadratically convergent phase, according to \cite{Andrei2022}, it takes
$	N_\mathrm{Q} = \left\lceil\log_2\left(
1-\log_2\left(
S^2L\rho
\right)
\right)\right\rceil$
iterations to obtain the solution to Problem (\ref{equ:barrier_problem}).
%According to \cite{Andrei2022}, it takes $N_\mathrm{N}$ iterations for backtracking line search, where $N_\mathrm{N}$ is given by~\cite[eqs. (4.45), (4.47), \& (4.52)]{Andrei2022}
Overall, the infeasible start Newton method takes $N_\mathrm{N}$ iterations per LP iteration:%, where $N_\mathrm{N}$ is given by
\begin{equation}
	\begin{aligned}
	N_\mathrm{N} &= N_lN_{\mathrm{D}} +N_{\mathrm{Q}}\\
	=&
	\left\lceil \dfrac{
		\log\left(S^2L\kappa
		\right)
	}
	{\log\left(\frac{1}{\tau}\right)} \!\!\right\rceil
	\left\lceil
	\dfrac{ S^2L\kappa}
	{\zeta\tau}
	\right\rceil\! +\!\left\lceil\log_2\left(1\!-\!\log_2\left(\!S^2\!L\!\rho\!\right)\right)\right\rceil.
	\end{aligned}
	\end{equation}

Moreover, the greedy method used for the discretization of user pairing in Alg. \ref{alg:user_pairing} takes $N_g$ iterations to obtain the discrete assignment strategy: %, where $N_g$ is given by
\begin{equation}
	N_g = \sum_{m=1}^{2K-1}\left(2K-m\right)=K\left(2K-1\right).
	\end{equation}
%{\color{black}Please refer to Lines 10 and 11 in Alg. \ref{alg:user_pairing}.}
% \begin{equation}
% 	\begin{aligned}
% 	N =& N_{\mathrm{N}}N_{\mathrm{LP}} + N_g\\
% 	=&\left[\left\lceil \dfrac{
% 		\log\left(S^2L\kappa
% 		\right)
% 	}
% 	{\log\left(\frac{1}{\tau}\right)} \right\rceil
% 	\left\lceil
% 	\dfrac{ S^2L\kappa}
% 	{\zeta\tau}
% 	\right\rceil +\left\lceil\log_2\left(1-\log_2\left(S^2L\rho\right)\right)\right\rceil\right]\\
% 	&\hspace{25mm}\times \left\lceil \frac{\log \left( \frac{K\left( 2K-1\right)  }{\epsilon t^{\left( 0\right)  }} \right)  }{\log \left( \xi \right)  } \right\rceil  + K\left(2K-1\right)
% 	\end{aligned}
% \end{equation}
\subsubsection{Convergence of Overall Algorithm}

We can interpret Alg. \ref{alg:user_pairing} as a mapping $\tilde{Q}$ from $\mathbf{\hat{X}}$ to $\mathbf{X}$, i.e., $\tilde{Q}: \mathbb{R}^{2K\times 2K}\to \mathbb{R}^{2K\times 2K}$.
In this case, the problem solved by Alg.~\ref{alg:overall}, i.e., Problem~\eqref{equ:original_problem}, can be rewritten as
\begin{subequations}
	\label{equ:newton_pairing_matrix_analyze_convergence}
	\begin{align}
		\label{equ:objective_pairing_relaxation_analyze_convergence}
		&\mathop{\mathrm{min}}_{x_{m,n}, \hat{x}_{m,n}, p_m, p_n}
		-2^{\mathrm{tr}\left( \mathbf{R}_{s}^{T}\mathbf{X} \right)}
		 \!=\!
		 -2^{\mathrm{tr}\left( \mathbf{R}_{s}^{T}\tilde{Q}\left( \hat{\mathbf{X}} \right) \right)}
		 \\
		\label{equ:rm_pairing_relaxation_analyze_convergence}
		\text{s.t. } & {R_{m,m}} \geqslant {\hat{x}_{m,n}} R_m,\\
		\label{equ:rn_pairing_relaxation_analyze_convergence}
		&  {R_{n,n}} \geqslant {\hat{x}_{m,n}} R_n,\\
		\label{equ:sum_power_pairing_relaxation_analyze_convergence}
		&  \sum\nolimits_{m = 1}^{2K} {\sum\nolimits_{n = m + 1}^{2K} {{\hat{x}_{m,n}}\left( {{p_n} + {p_m}} \right)} }  \leqslant P,\\
		\label{equ:integer_pairing_relaxation_analyze_convergence}
		&  0 \leqslant {\hat{x}_{m,n}}, \leqslant 1,1 \leqslant m,n \leqslant 2K,\\
		\label{equ:xmm_0_pairing_relaxation_analyze_convergence}
		&   {\hat{x}_{m,n}} = 0, \\
		\label{equ:sum_m_pairing_relaxation_analyze_convergence}
		&   \sum\nolimits_{m = 1}^{2K} {\hat{x}_{m,n}}  = 1,1 \leqslant n \leqslant 2K,  \\
		\label{equ:last_constraints_analyze_convergence}
		& \sum\nolimits_{n = 1}^{2K} {\hat{x}_{m,n}}  = 1,1 \leqslant m \leqslant 2K,
	\end{align}
\end{subequations}
where $\mathbf{R}_s\in \mathbb{R}^{2K\times 2K}$ is the secrecy rate matrix whose $(m, n)$-th element is the secrecy rate of user $n$ against the potential eavesdropping by user $m$.

We can further interpret Alg. \ref{alg:overall} as a mapping $Q: \mathbb{R}^{2K} \times \mathbb{R}^{2K\times 2K} \times \mathbb{R}^{2K\times 2K} \to \mathbb{R}$, which maximizes the sum secrecy rate. Then,
\begin{equation}
	\label{equ:objective_overll_algorithm}
	Q\left(\bar{\mathbf{p}},  \mathbf{\hat{X}}, \mathbf{X}\right) = -2^{\mathrm{tr}\left( \mathbf{R}_{s}^{T}\tilde{Q}\left( \hat{\mathbf{X}} \right) \right)}+\delta_{\mathcal{F}}\left(\bar{\mathbf{p}},  \mathbf{\hat{X}}, \mathbf{X}\right)
	\end{equation}
where $\mathbf{\bar{p}}=\left\{p_n,\forall n\right\}$ is the vector of the transmit powers;
and the indicator function $\delta_{\mathcal{F}}\left(\bar{\mathbf{p}},  \mathbf{\hat{X}}, \mathbf{X}\right)$ is given by
\begin{equation}
	\delta_{\mathcal{F}}\left(\bar{\mathbf{p}},  \mathbf{\hat{X}}, \mathbf{X}\right) =
	\begin{cases}
		0, \hspace{3.8mm}\text{ if }
		\left(\bar{\mathbf{p}},  \mathbf{\hat{X}}, \mathbf{X}\right) \in \mathcal{F}\\
		+\infty, \hspace{0.5mm}\mathrm{otherwise}.
		\end{cases}
	\end{equation}
Here, $\mathcal{F}$ is the feasible domain of \eqref{equ:newton_pairing_matrix_analyze_convergence} defined by \eqref{equ:rm_pairing_relaxation_analyze_convergence}--\eqref{equ:last_constraints_analyze_convergence}.

As a result, Alg. \ref{alg:overall} can be interpreted to solve  \eqref{equ:newton_pairing_matrix_analyze_convergence} using the Block Coordinate Descent (BCD). In each iteration, the algorithm sequentially solves subproblems
% \begin{subequations}
% 	\label{equ:solve_problems_each_iter}
% 	\begin{equation}
% 		\label{equ:argmin_p}
		$\min_{\mathbf{\bar{p}}}Q\left(\bar{\mathbf{p}},  \mathbf{\hat{X}}, \mathbf{X}\right)$,
% 		\end{equation}
% 	\begin{equation}
% 		\label{equ:argmin_hatX}
		$\min_{\mathbf{\hat{X}}}Q\left(\bar{\mathbf{p}},  \mathbf{\hat{X}}, \mathbf{X}\right)$,
% 	\end{equation}
and
% \begin{equation}
% 	\label{equ:argmin_X}
	$\min_{\mathbf{X}}Q\left(\bar{\mathbf{p}},  \mathbf{\hat{X}}, \mathbf{X}\right)$.
% \end{equation}
% 	\end{subequations}
The convergence of each of the subproblems has been confirmed, since
Section~\ref{sec:power_allocation_opt} shows the semi-closed solution $\bar{\textbf{p}}$ and Section~\ref{sec:convergence_user_pairing} analyzes the convergence of
${\mathbf{\hat{X}}}$
%\eqref{equ:argmin_hatX}
and
%\eqref{equ:argmin_X}.
${\mathbf{X}}$. The overall convergence rate of Alg. \ref{alg:overall} is established in the following, starting with a few definitions.
\begin{definition}[Semi-algebraic set\cite{Lasserre2015,Sun2018,Bolte2014}]
%	\paragraph{Semi-algebraic set\cite{Lasserre2015,Sun2018,Bolte2014}}
	A subset of $\mathbb{R}^n$, denoted by $\mathcal{D}$, is called semi-algebraic if there exists finite $U, V\in \mathbb{N}$, such that
	\begin{equation}
		\label{equ:define_semi_algebraic}
		\mathcal{D}=\cup_{u=1}^U\cap_{v=1}^V\Big\{\mathbf{z}\in \mathbb{R}^{n}\left|p_{u, v}\left(\mathbf{z}\right)=0, q_{u, v}\left(\mathbf{z}\right)>0\Big\}\right.
	\end{equation}
	where $p_{u, v}\left(\mathbf{z}\right)$ and $q_{u, v}\left(\mathbf{z}\right)$ are real polynomial functions for $u =1, \cdots, U$ and $v=1, \cdots, V$.
	\end{definition}
\begin{definition}[Semi-algebraic function\cite{Lee2016,Attouch2010,Bao2014}]
%	\paragraph{Semi-algebraic function\cite{Lee2016,Attouch2010,Bao2014}}
	Let $\mathcal{D}\in \mathbb{R}^{n}$ and $\mathcal{E}\in \mathbb{R}^{m}$ be two semi-algebraic sets. A mapping $F:\mathcal{D}\to \mathcal{E}$ is semi-algebraic if its graph
	\begin{equation}
		\left\{\left(\mathbf{z}, \mathbf{o}\right)\in \mathcal{D}\times \mathcal{E}\left|\mathbf{o} = F\left(\mathbf{z}\right)\right\}\right.
		\subset \mathbb{R}^n\times \mathbb{R}^m
	\end{equation}
	is a semi-algebraic set.
\end{definition}

%We first propose the Theorem \ref{theorem:power_allocation_semialgebraic} as follows.
\begin{lemma}
	\label{theorem:power_allocation_semialgebraic}
	The function $Q(\cdot)$ is semi-algebraic.
\end{lemma}
\begin{proof}
Please refer to  Appendix \ref{appendix:proof_theorem1}.
\end{proof}

With the aid of Lemma \ref{theorem:power_allocation_semialgebraic},  the convergence rate of Alg.~\ref{alg:overall} can be established:
\begin{theorem}
	\label{theorem:convergence_rate}
	When $S$ and $L$ exists, there exist constants $C, \varrho, q_0 > 0$, satisfying the following inequality
	\begin{equation}
		\label{equ:convergence_rate}
		\eta \leqslant C q^{-\frac{1}{\varrho}}
		\end{equation}
	after $q > q_0$ iterations of the overall algorithm, where $\eta>0$ is the tolerance. In other words, $q \sim  \mathcal{O}\left(\dfrac{1}{\eta^{\varrho}}\right)$.
	\end{theorem}
\begin{proof}
	See Appendix \ref{appendix:proof_theorem2}.
	\end{proof}

%{\color{black} \textbf{COMMENT:} Please first confirm the proposed algorithm converges, e.g., the optimization objective is monotonic during the alternating optimization of power allocation and user pairing. Also, please define S and L here. I believe they are relevant to the confirmation of the convergence of Algo. 1.}

\subsection{Complexity Analysis}
\subsubsection{Power Allocation}
Since \eqref{equ:power_allocation_pn_result} and \eqref{equ:pmn_results}
%and \eqref{equ:solve_the_semi_closed_solution}
provide the closed-form power allocation strategy per user group, the complexity, denoted by $\mathcal{T}_{\mathrm{PA}}$, depends linearly on the number of user groups, $K$; i.e.,
% \begin{equation}
% 	\label{equ:power_allocation_complexity}
$	\mathcal{T}_{\mathrm{PA}} = \mathcal{O}\left(K\right)$.
% \end{equation}
%FLOPs to obtain allocated power of all users.

\subsubsection{User Pairing}
We analyze the computational complexity of solving (\ref{equ:newton_pairing_matrix}) using the logarithmic barrier method.
As discussed in Sec.~\ref{sec:convergence_user_pairing}, the desired accuracy $\epsilon$ is achieved after $N_{\mathrm{LP}}$ logarithmic barrier method iterations. In each of the iterations, the infeasible start Newton method is performed.

%Next, we analyze the complexity per iteration, i.e., for solving (\ref{equ:barrier_problem}).
The infeasible start Newton method also iterates. In each iteration of the damped Newton phase of the infeasible start Newton method,
computing (\ref{equ:primal_newton_step}) through the LU decomposition takes $\mathcal{T}_{\mathrm{LU}}= \frac{2}{3} \left[K\left(2K+1\right)\right]^3 + 2\left[K\left(2K+1\right)\right]$ floating operator points (FLOPs)~\cite{Golub2013}.
The complexity of the backtracking line search is
%\begin{equation}
	$\mathcal{T}_1 =\mathcal{O}\left(K^2\right)$ per step.
%\end{equation}
As a result, the backtracking line search in the damped Newton phase is
%\begin{equation}
$	\label{equ:ts}
	\mathcal{T}_s \!=\! N_l\mathcal{T}_1 \!=\!
	\mathcal{O}
\left(\log\left(S^2L\kappa
	\right)
 K^2/\log\left(\frac{1}{\tau}\right)\right)$.
%\end{equation}
Moreover, updating {\color{black}$\hat{\mathbf{x}}$}
%$\hcancel[red]{\mathbf{x}}$
and $\mathbf{w}$
in Line 6 of Alg. \ref{alg:user_pairing} incurs
%\begin{align}
%	\label{equ:tx}
	$\mathcal{T}_\mathbf{\hat{x}} = 2K(2K-1) $ and
%	\label{equ:tw}
	$\mathcal{T}_\mathbf{w} =  4K$
%	\end{align}
FLOPs~\cite{Zhang2017}. Hence, the complexity of the damped Newton phase is
\begin{equation}
\begin{aligned}
	\mathcal{T}_\mathrm{D} &= N_\mathrm{D}\left({\color{black}\mathcal{T}_{\mathrm{LU}}+}
	\mathcal{T}_s + \mathcal{T}_{\mathbf{\hat{x}}}+\mathcal{T}_{\mathbf{w}}
	\right) \\
	&=
	\mathcal{O}\left(
	\dfrac{
		S^2L\kappa
		\log\left(
		S^2L\kappa
		\right)
	}{
		\zeta\tau
		\log\left(\frac{1}{\tau}\right)
	}K^2 {\color{black}+\frac{16}{3}K^6}
	\right).
	\end{aligned}
	\end{equation}
Likewise, the complexity of the quadratically convergent phase is given by
\begin{equation}
\begin{aligned}
	\mathcal{T}_\mathrm{Q}&=N_\mathrm{Q}\left({\color{black}\mathcal{T}_{\mathrm{LU}}+}\mathcal{T}_{\mathbf{\hat{x}}} + \mathcal{T}_{\mathbf{w}}\right) \\
	&=
	\mathcal{O}
	\left(
	\log\left(
	1-\log\left(
	S^2L\rho
	\right)
	\right) K^2 {\color{black}+\frac{16}{3}K^6}
	\right)
	\end{aligned}
	\end{equation}
Thus, the complexity of the logarithmic barrier method is
\begin{equation}
	\begin{aligned}
		\label{equ:LP_complexity}
	\mathcal{T}_{\mathrm{LP}}=&N_{\mathrm{LP}}\left( \mathcal{T}_{\mathrm{D}}+\mathcal{T} _{\mathrm{Q}} \right) \\
=&
\mathcal{O} \left(
\frac{\log \left( \frac{K^2}{\epsilon t^{\left( 0 \right)}} \right)}{\log \left( \xi \right)}
\left( \frac{S^2L\kappa \log \left( S^2L\kappa \right)}{\zeta \tau \log \left( \frac{1}{\tau} \right)}+ \right. \right. \\
&\left.\log \left( 1-\log \left( S^2L\rho \right) \right) \Bigg)K^2{\color{black}+
\dfrac{32}{3}K^6}\right)
	\end{aligned}
	\end{equation}
%\subsubsection{Greedy-based approach}

The greedy method for discretization of the user pairing utilizes double loops to search for discrete user pairing strategies. Thus, the computational complexity is
$	\mathcal{T}_{g} = \mathcal{O}\left(K^2\right)$.
The overall complexity of user pairing in Alg. \ref{alg:user_pairing} is   $\mathcal{T}_{\mathrm{UP}} = \mathcal{T}_{\mathrm{LP}} + \mathcal{T}_g $.

\subsubsection{Overall Complexity}
According to Theorem \ref{theorem:convergence_rate}, % and \eqref{equ:number_of_iteration_of_overall_algorithm},
it takes $q=\mathcal{O}\left(\dfrac{1}{\eta^{\varrho}}\right)$ iterations for Alg. \ref{alg:overall} to converge.
Therefore, the overall complexity is
% \begin{align}
% %	\begin{aligned}
% 		\label{equ:complexity1}
$\mathcal{T} = \left(\mathcal{T}_{\mathrm{PA}} + \mathcal{T}_{\mathrm{UP}}\right) q$.
% 		&=
% 		\mathcal{O} \left( \frac{K^2\log \left( \frac{K^2}{\epsilon t^{\left( 0 \right)}} \right)}{\eta ^{\varrho}\log \left( \xi \right)} \right. \\
% 		&\hspace{10mm}
% 		\times \left[ \left. \left. \frac{S^2L\kappa \log \left( S^2L\kappa \right)}{\zeta \tau \log \left( \frac{1}{\tau} \right)}+\log \left( 1-\log \left( S^2L\rho \right) \right) \right] \right) \right. \nonumber
% %	\end{aligned}
% 	\end{align}

%{\color{black}Please provide a paragraph to discuss the complexity and convergence of the algorithm and, if possible, the optimality (e.g., in each of the two steps).}

\ifCLASSOPTIONcaptionsoff
  \newpage
\fi

\section{Simulation and Discussion}
Extensive simulations are provided to gauge the proposed scheme, where the users are distributed in a circular area with radius $l = 300$ m and the BS is located at the center of the area. The path loss exponent is set to 3. The bandwidth of each resource block is 0.5 MHz. The receiver noise power spectral density is $-174$ dBm/Hz.

To assess the merits of the proposed algorithm, we compare the algorithm with the following alternative approaches:
\begin{itemize}
    \item
\textit{Equal power allocation (EPA)}: This mechanism allocates the same transmit power for all users. The proposed user pairing is used.
%This allows us to evaluate the performance of the user pairing algorithm independently of the power allocation strategy.
By comparing our algorithm with the EPA, we can assess the benefit of the proposed power allocation strategy.
    \item
\textit{Random pairing (RP):} This mechanism randomly selects user pairs. We use the optimal power allocation strategy proposed in this paper to determine the transmit power for each user.
 \item
 \textit{Gale-Shapley algorithm-based alternative:} We take the Gale-Shapley algorithm~\cite{Zhou2017} to pair users, without considering their channel conditions. The optimal power allocation strategy proposed in this paper is used to determine the transmit power for each user.
 \item
 \textit{Simplex method-based alternative:} We take the Simplex method~\cite{Ficken2015} to solve the LP problem in user pairing. The optimal power allocation strategy developed in this paper is adopted to specify the transmit power for each user.
 \end{itemize}
By comparing the proposed algorithm with the RP-, Gale-Shapley algorithm- and Simplex method-based alternatives, we can evaluate the gain of the proposed user pairing algorithm.

\begin{figure}%[!htb]
    \centering
    \includegraphics[width=1\columnwidth]{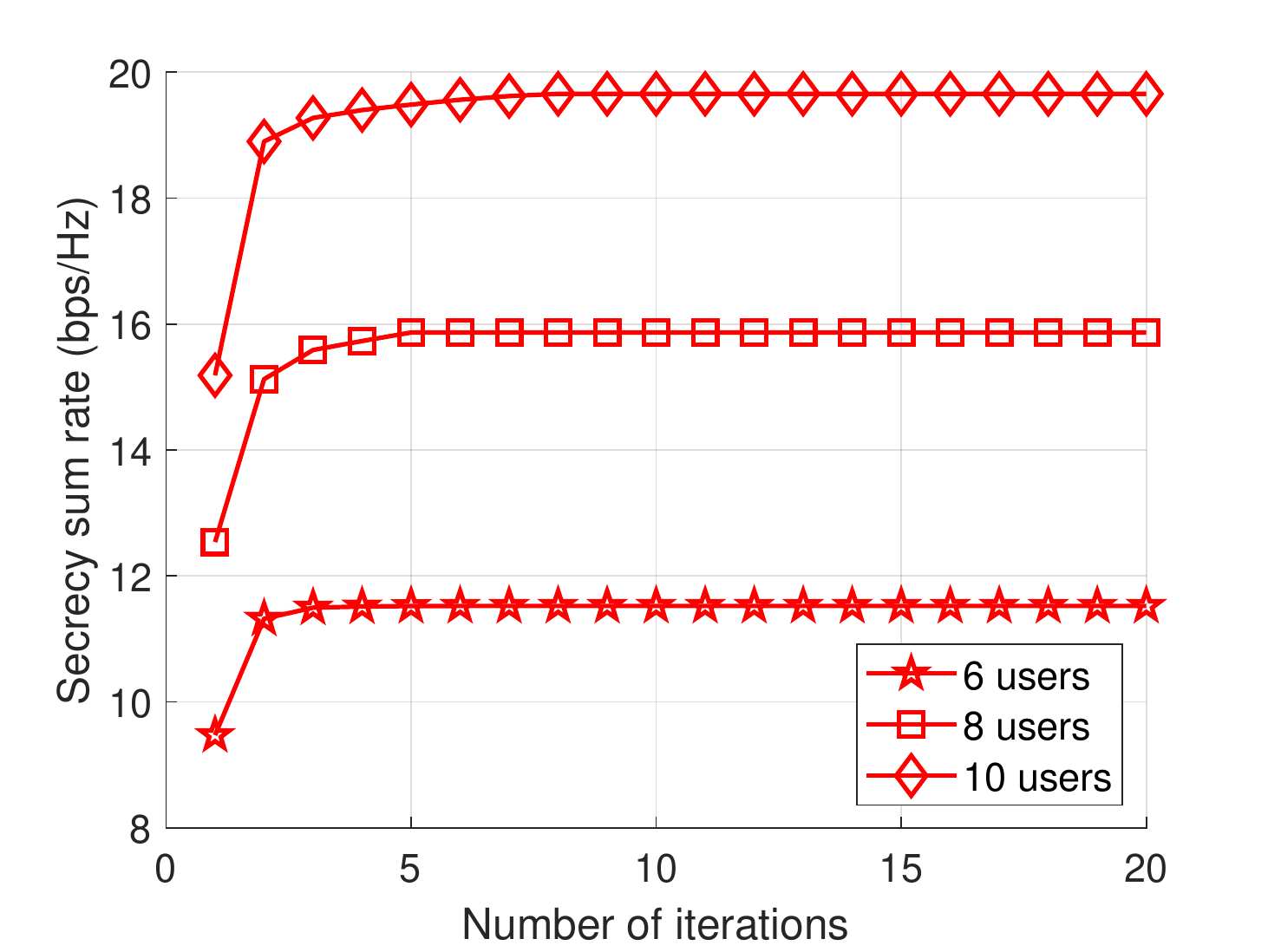}
    \vspace*{0pt}
    \caption{Sum secrecy rate against number of iterations with $2K$ = 6, 8, 10, and ${P} = 20$ dBm.}
    \label{fig:num_iters}
    \end {figure}

Fig. \ref{fig:num_iters} shows the evolution of the sum secrecy rate as the number of iterations increases. It is observed that the sum secrecy rate rises quickly and usually converges within 10 iterations.
%This indicates that the proposed algorithm is able to quickly find a good solution to the optimization problem and achieve a high sum secrecy rate.
It is also noticed that the user number has a non-negligible impact on the sum secrecy rate, especially when there are many users. The reason is that more users lead to stronger interference, hence penalizing the sum secrecy rate. Our algorithm mitigates the interference by properly allocating the power and pairing the users. These observations highlight the importance of the proposed algorithm in multiuser NOMA systems with many users.

\begin{figure}%[!htb]
	\centering
	\includegraphics[width=1\columnwidth]{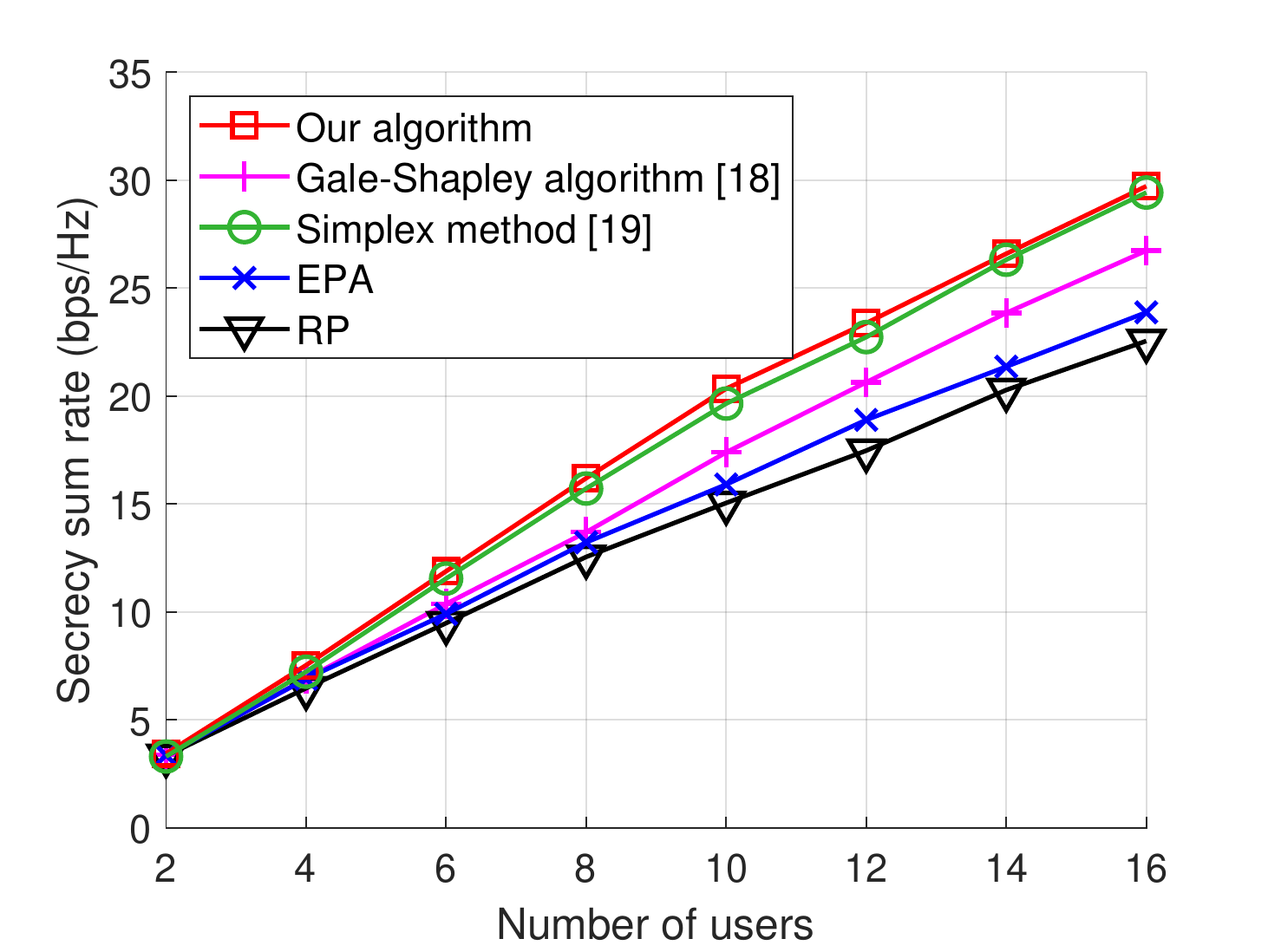}
	\vspace*{0pt}
	\caption{Sum secrecy rate against user number with ${P} = 20$ dBm.}
    \label{fig:number_user}
	\end {figure}

Fig. \ref{fig:number_user} plots the sum secrecy rate as users increase in the considered system. We notice that the sum secrecy rate increases with users under all five schemes. Our approach consistently outperforms the benchmark schemes, EPA and RP, indicating the algorithm can effectively allocate the powers and pair the users to promote the sum secrecy rate.
%This confirms the effectiveness of the proposed algorithm.
In order to verify the user pairing algorithm delivered in this paper, we compare our algorithm with the Gale-Shapley and Simplex methods. We observe that the new user pairing is more effective than the Gale-Shapley and Simplex algorithms. Although the gap of the sum secrecy rate is small between the proposed algorithm and the Simplex method, our algorithm is significantly more computationally efficient than the Simplex method, as will be discussed shortly.

\begin{figure}%[!htb]
	\centering
	\includegraphics[width=1\columnwidth]{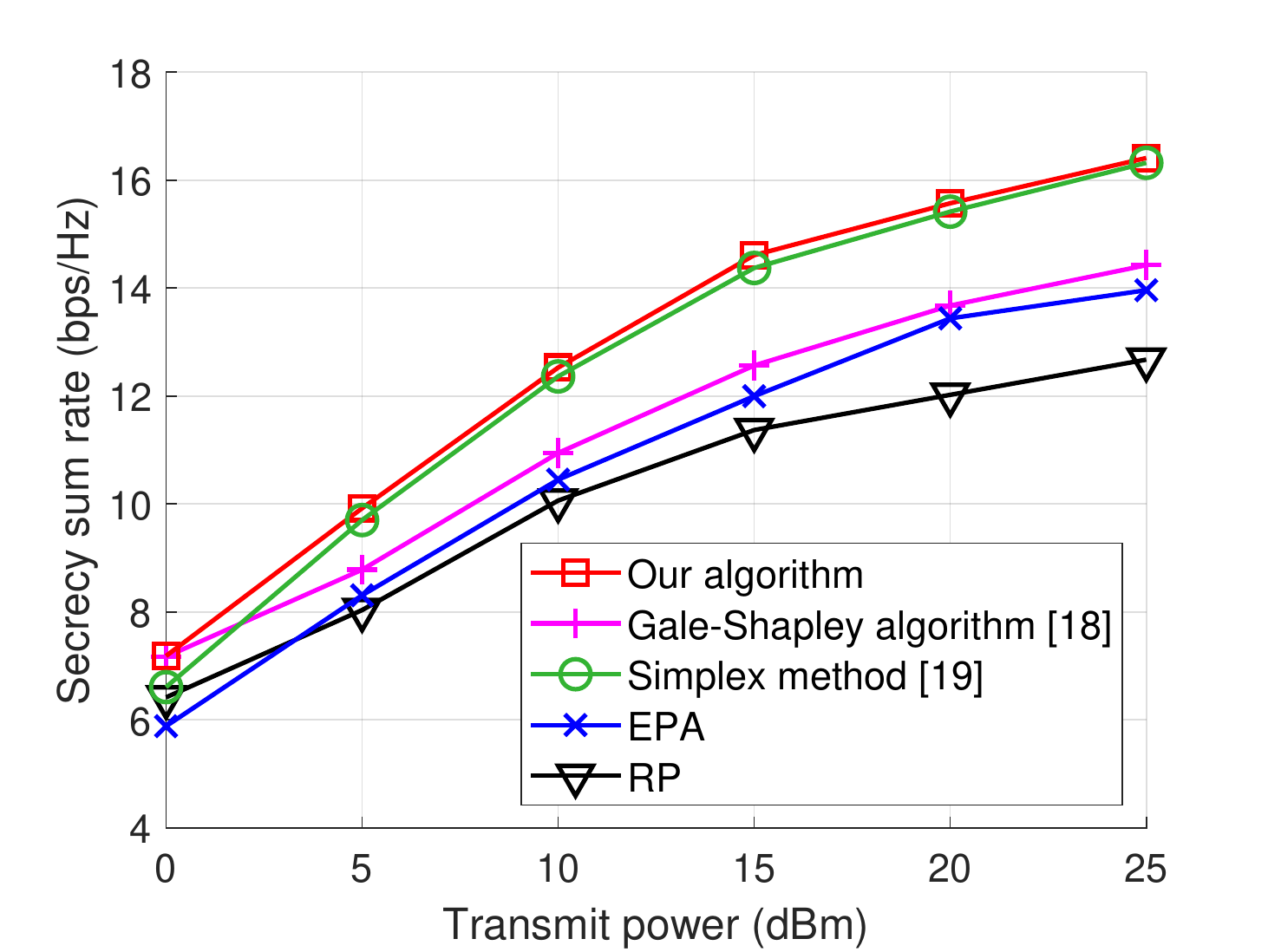}
	\vspace*{0pt}
	\caption{Sum secrecy rate against transmit power with $2K = 8$.}
    \label{fig:transmit_power}
	\end {figure}

\begin{figure}%[!htb]
	\centering
	\includegraphics[width=1\columnwidth]{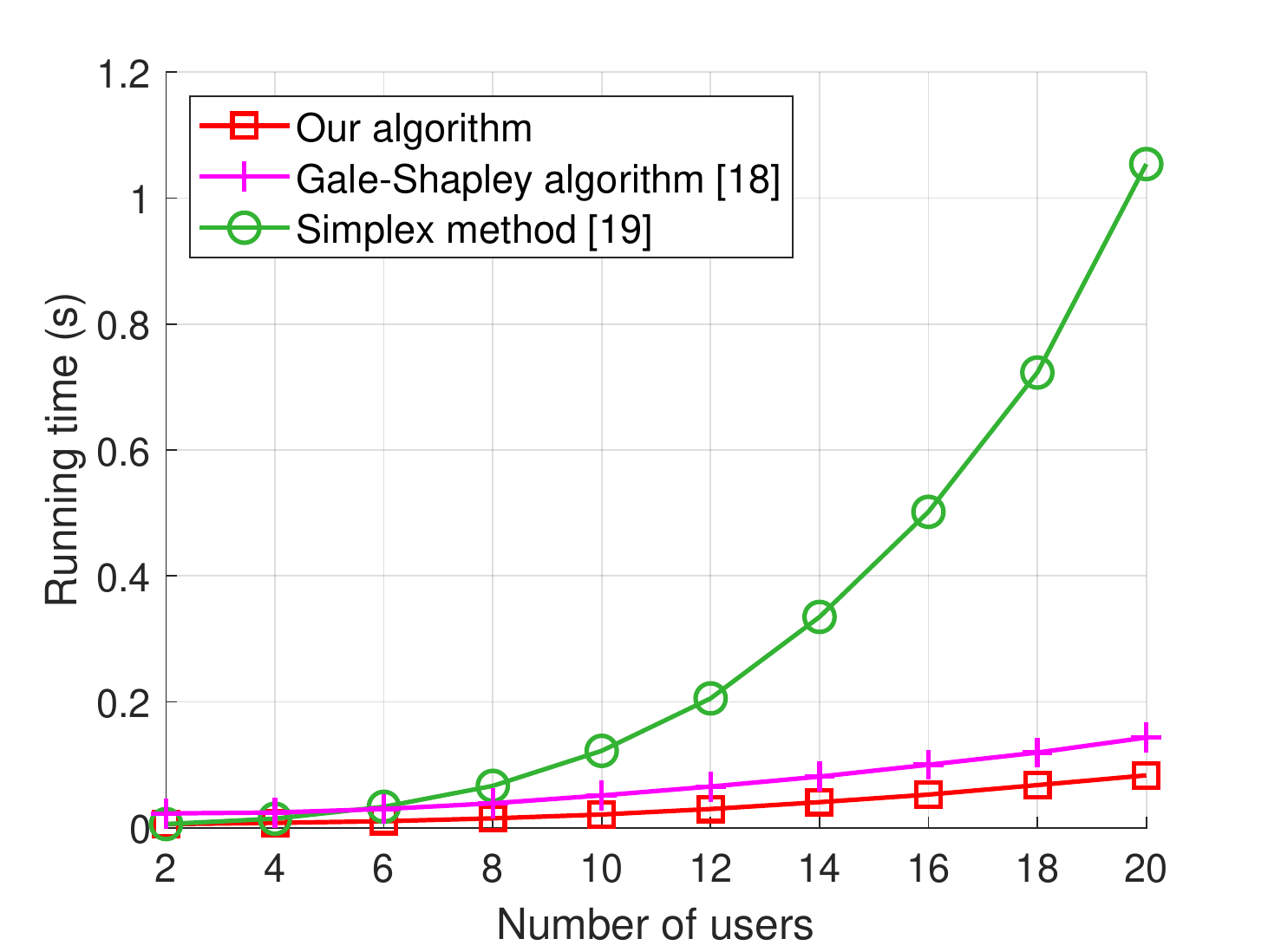}
	\vspace*{0pt}
	\caption{Running time against user number with ${P} = 20$ dBm.}
    \label{fig:running_time}
	\end {figure}

\begin{figure}[!htb]
	\centering
	\includegraphics[width=1\columnwidth]{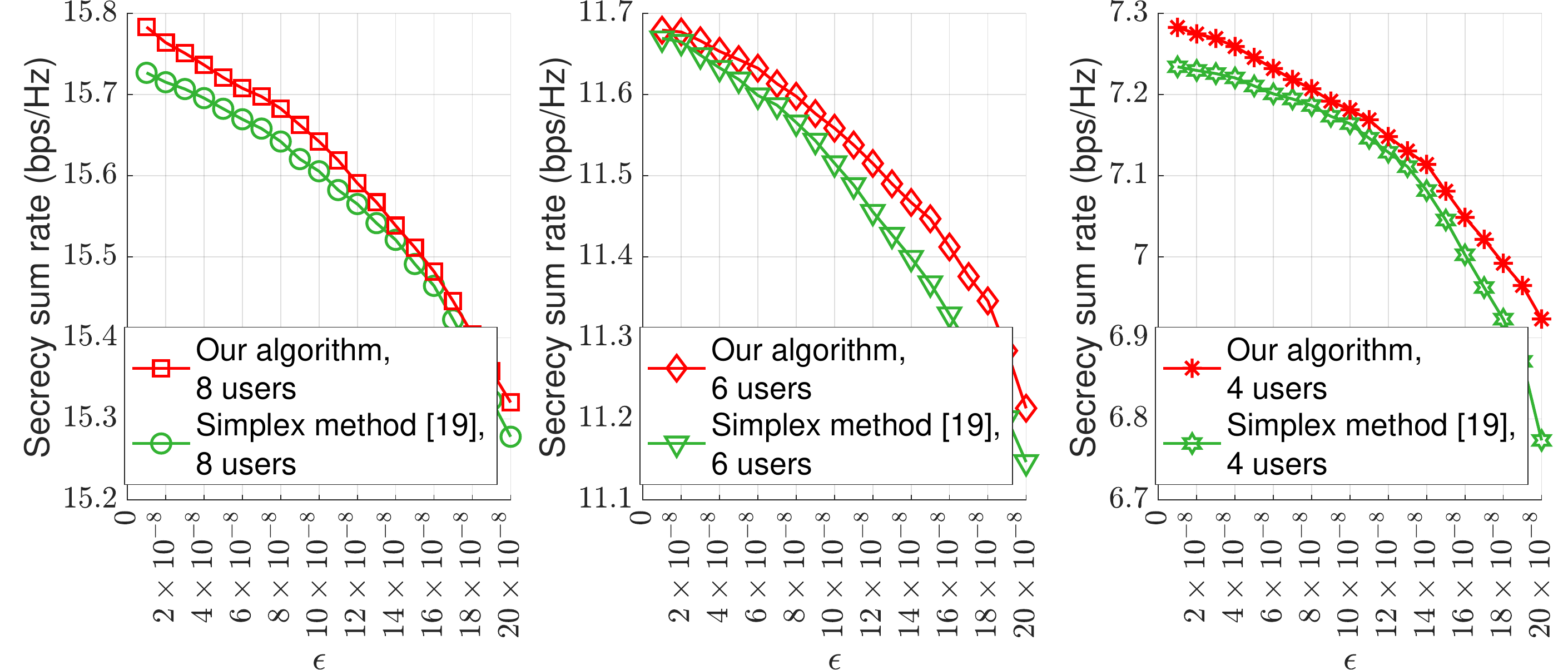}
	\vspace*{0pt}
	\caption{Sum secrecy rate  against $\epsilon$ with ${P} = 20 $ dBm.}
	\label{fig:epsilon}
\end {figure}

Fig. \ref{fig:transmit_power} shows that the sum secrecy rate increases with the transmit power of the BS. This is expected because more transmit power means that the users can transmit their messages at higher levels, which improves the secrecy performance. We also observe that our algorithm consistently outperforms the benchmarks, EPA, RP, Gale-Shapley and Simplex based algorithm. In other words, the new algorithm can effectively allocate the power and pair the users, leading to improved sum secrecy rate compared to alternative methods. In addition, we see that the NOMA-based EPA outperforms the NOMA-based RP, suggesting that NOMA systems are more sensitive to user pairing than they are to power allocation. All this confirms the effectiveness of our algorithm and the criticality of user pairing in NOMA.
%However, the sum secrecy rate also decreases with the number of users, which is due to the increased interference that arises when there are more users in the system.

% Fig. \ref{fig:running_time} shows the running time as a function of the numbers of users in the system. We can see that the running time increases with the number of users, which is expected because more users means that the system needs more time to process user pairing, which increases the running time. We also note that our proposed algorithm in this paper has the lowest complexity and the simplex based algorithm has the highest complexity when they achieve the same sum secrecy rate.
% Fig. \ref{fig:epsilon} illustrates the relationship between $\epsilon$ in Alg. \ref{alg:user_pairing} and sum secrecy rate with various number of users. We can find that the larger $\epsilon$ will decrease the sum secrecy rate, since $\epsilon$ is actually the tolerance of errors. However, smaller $\epsilon$ will suffer from larger computational complexity with more iterations. To achieve the balance of sum secrecy rate and running time, we choose $\epsilon=1\times 10^8$.

Fig. \ref{fig:running_time} demonstrates the relationship between the running time and user number in the considered system. As anticipated, the running time increases with users, since a larger number of users require more time for the system to decide on user pairing, extending considerably the running time. It is noticed that our algorithm has the lowest complexity and the Simplex method-based alternative has the highest, albeit they achieve similar sum secrecy rates.

% Fig. \ref{fig:epsilon} illustrates the impact of $\epsilon$ in Alg. \ref{alg:user_pairing} on the sum secrecy rate for different numbers of users. It is observed that a larger $\epsilon$ value leads to a lower sum secrecy rate since $\epsilon$ represents the tolerance for errors. On the other hand, a smaller $\epsilon$ value results in a higher computational complexity with more iterations. We choose $\epsilon=1\times 10^8$ to strike a balance between the sum secrecy rate and running time.

Last but not least, Fig.~\ref{fig:epsilon} presents the effect of the parameter $\epsilon$ on the sum secrecy rate for various numbers of users in Alg.~\ref{alg:user_pairing}. It is noticed that a larger $\epsilon$ value can cause a faster decrease in the sum secrecy rate, as $\epsilon$ represents the tolerance for errors in the proposed algorithm. However, a smaller $\epsilon$ value can result in a higher computational complexity, requiring more iterations to reach a satisfactory solution. Here, we set $\epsilon=1\times 10^8$ in order to balance the trade-off between sum secrecy rate and running time.

%{\color{black}Can you produce more results, e.g., under different distributions of the users?}

\section{Conclusion}
This paper studied a multiuser NOMA system with untrusted IoT devices, and drew up a joint power allocation and user pairing problem in an attempt to achieve the maximum sum secrecy rate of the system, subject to the data rate requirements of individual users and the transmit power of the BS. To effectively solve this MINLP problem, we decomposed the problem between two subproblems: power allocation with a closed-form solution, and user pairing solved using the logarithmic barrier method. Simulations showed that our algorithm offers superior secrecy performance to existing alternatives, indicating that the algorithm is effective in improving the secrecy of NOMA systems and that holistic consideration of both user pairing and power allocation is critical.

\bibliographystyle{IEEEtran}
\bibliography{reference}

\appendices
\section{Proof of Lemma \ref{theorem:power_allocation_semialgebraic}}
\label{appendix:proof_theorem1}
To prove Lemma \ref{theorem:power_allocation_semialgebraic}, we first %give the Lemma \ref{lemma:greedy_algorithm_semi_algebraic} as follows.
%\begin{lemma}
%	\label{lemma:greedy_algorithm_semi_algebraic}
%	The
prove the mapping $\tilde{Q}\left(\mathbf{\hat{X}}\right)=\mathbf{X}$ is semi-algebraic.
%	\end{lemma}
%\begin{proof}
To do this, we define $\mathbf{U}=\left\{\left(m_1, n_1\right), \cdots, \left(m_K, n_K\right)\right\}$ to collect the indices chosen by the greedy method in Alg. \ref{alg:user_pairing}, and
	\begin{subequations}
	\begin{equation}
		x_{m_1, n_1} \geqslant \cdots \geqslant x_{m_K, n_K};
		\end{equation}
	\begin{equation}
		x_{m_k, n_k} \geqslant x_{m', n'}, \forall \left(m_k, n_k\right), \left(m', n'\right)\notin \mathbf{U}.
		\end{equation}
	\end{subequations}
Let $\mathcal{U}$ be the set of all possible $\mathbf{U}$. The graph of $\tilde{Q}\left(\mathbf{\hat{X}}\right)=\mathbf{X}$ can be given by
\begin{align}
	\label{equ:graph_greedy_algorithm}
%	\begin{aligned}
	\mathrm{graph}&\left(\tilde{Q}\right) =
	\bigcup_{\mathbf{U}\in \mathcal{U}}{\left\{ \left. \left( \hat{\mathbf{X}},\mathbf{X} \right) \right| \right.}{x}_{m,n}=1,{x}_{m^\prime,n^\prime}=0, \nonumber\\
	& \forall \left( m,n \right) \in \mathbf{U},\left( m^\prime,n^\prime \right) \notin \mathbf{U}\Big\}\\
	=& \bigcup_{\mathbf{U}\in \mathcal{U}}{\left[ \left( \bigcap_{\left( m,n \right) \in \mathbf{U}}{\left\{ \left. \left( \hat{\mathbf{X}},\mathbf{X} \right) \right|x_{m,n}=1 \right\}} \right) \right.}\nonumber \\
	& \left. \bigcap{\left( \bigcap_{\left( m^\prime,n^\prime \right) \notin \mathbf{U}}{\left\{ \left. \left( \hat{\mathbf{X}},\mathbf{X} \right) \right|x_{m^\prime,n^\prime}=0 \right\}} \right)} \right]. \nonumber
%	\end{aligned}
	\end{align}
Since \eqref{equ:graph_greedy_algorithm} is semi-algebraic, $\tilde{Q}\left(\mathbf{\hat{X}}\right)=\mathbf{X}$ is semi-algebraic.
%	\end{proof}

Next, %Based on Lemma \ref{lemma:greedy_algorithm_semi_algebraic}, we give the Lemma \ref{lemma:objective_semi_algebraic} as follows.
%\begin{lemma}
%	\label{lemma:objective_semi_algebraic}
we prove that the objective function of
	\eqref{equ:newton_pairing_matrix_analyze_convergence} is semi-algebraic.
%	\end{lemma}
%\begin{proof}
	The graph of \eqref{equ:objective_pairing_relaxation_analyze_convergence} is
	\begin{align}
		\label{equ:graph_of_objective}
%		\begin{aligned}
		&\left\{ \left( \bar{\mathbf{p}}, \hat{\mathbf{X}}, \mathbf{X}, z \right) \left| -2^{\mathrm{tr}\left( \!\mathbf{R}_{s}^{T}\hat{\mathbf{X}} \right)}=z \right. \right\} \nonumber\\
		=&
		\bigcup_{\mathbf{U}\in \mathcal{U}}{\left[ \left\{ \left( \bar{\mathbf{p}},\hat{\mathbf{X}},\mathbf{X},z \right) \left| \prod_{\left( m,n \right) \in \mathbf{U}}{\check{p}_{m,n}}=\sigma ^2z\prod_{\left( m,n \right) \in \mathbf{U}}{\check{q}_{m,n}} \right. \right\} \right.} \nonumber\\
		&
		\bigcap{\left( \bigcap_{\left( m,n \right) \in \mathbf{U}}{\left\{ \left. \left( \bar{\mathbf{p}},\hat{\mathbf{X}},\mathbf{X},z \right) \right|x_{m,n}=1 \right\}} \right)}\\
		&
		\bigcap{\left( \bigcap_{\left( m^{\prime},n^{\prime} \right) \notin \mathbf{U}}{\left\{ \left. \left( \bar{\mathbf{p}},\hat{\mathbf{X}},\mathbf{X},z \right) \right|x_{m^{\prime},n^{\prime}}=0 \right\}} \right)} \nonumber\\
		&
		\left. \bigcap{\left( \bigcap_{\left( m,n \right) \in \mathbf{U}}{\left\{ \left. \left( \bar{\mathbf{p}},\hat{\mathbf{X}},\mathbf{X},z \right) \right|\check{p}_{m,n}-\sigma ^2\check{q}_{m,n}\geqslant 0 \right\}} \right)} \right], \nonumber
%		\end{aligned}
	\end{align}
where $
\check{p}_{m,n}
$
and $
\check{q}_{m,n}
$ are given by
\begin{subequations}
	\begin{equation}
		\check{p}_{m,n}=\left( \left| h_m \right|^2p_n+\sigma ^2 \right) \left( \left| h_n \right|^2p_n+\sigma^2 \right) ;
		\end{equation}
	\begin{equation}
		\check{q}_{m,n}=\left| h_m \right|^2\left( p_n+p_m \right) +\sigma^2.
		\end{equation}
	\end{subequations}
According to \cite[Corol. 4]{Neyman2003}, the composition of semi-algebraic function is semi-algebraic. %Therefore, based on Lemma \ref{lemma:greedy_algorithm_semi_algebraic},
Since the mapping $\tilde{Q}\left(\mathbf{\hat{X}}\right)=\mathbf{X}$ is semi-algebraic,
both \eqref{equ:graph_of_objective} and the objective function of \eqref{equ:newton_pairing_matrix_analyze_convergence} are semi-algebraic.
%	\end{proof}

%Then we propose the Lemma \ref{lemma:indicator_semi_algebraic} as follows
% \begin{lemma}
% 	\label{lemma:indicator_semi_algebraic}
Further, we prove that the indicator function $\delta_{\mathcal{F}}\left(\bar{\mathbf{p}}, \mathbf{\hat{X}}, \mathbf{X}\right)$ is semi-algebraic. Specifically,
% 	\end{lemma}
%\begin{proof}
	we prove that the feasible domain $\mathcal{F}$ is semi-algebraic as follows. We first show that the feasible domain $\mathcal{F}_1$ defined by \eqref{equ:rm_pairing_relaxation_analyze_convergence} is semi-algebraic. The set $\mathcal{F}_1$ is given by
%	\begin{figure*}[b]
%	\begin{strip}
	\begin{align}
		\label{equ:semi_algebraic_F1}
%		\begin{aligned}
		\mathcal{F}_1
		=&\bigcup_{\mathbf{U}\in \mathcal{U}}\left.\left\{\left(\mathbf{\bar{p}}, \mathbf{\hat{X}}, \mathbf{X}\right)\right|
		R_{m,m}\! \geqslant\! R_m, \!\hat{x}_{m,n}\!=\!1, \!\hat{x}_{m', n'}\!=\!0
		\right\} \nonumber\\
		=&\bigcup_{\mathbf{U}\in \mathcal{U}}\left[\!\left(\!\bigcap_{(m, n)\in \mathbf{U}} \right.\left\{\left(\mathbf{\bar{p}}, \mathbf{\hat{X}}, \mathbf{X}\right) \!\Big|
		\hat{p}_{m,n} \!\leqslant\! 0, \hat{x}_{m,n}\!=\!1
		\right\}\right) \nonumber\\
		&\bigcap
		\left.\!\left(\!\bigcap_{\left(m', n'\right)\notin \mathbf{U}}\left\{\left(\mathbf{\bar{p}}, \mathbf{\hat{X}}, \mathbf{X}\right) \! \Big|
		\hat{p}_{m',n'} \!\leqslant \! 0, \hat{x}_{m',n'}\!=\!0
		\right\}\!\right)\!\right] ,
%		\end{aligned}
		\end{align}
%	\end{strip}
%\end{figure*}
where $\hat{p}_{m,n}$ is given by
\begin{equation}
	\hat{p}_{m,n} = \left[\left(p_m+p_n\right)\!+\!\dfrac{\sigma^2}{\left|h_m\right|^2}\right]\!\left[p_n^2\!-\!\dfrac{\sigma^2}{\left|h_m\right|^2}\left(p_m\!-\!p_n\right)\right].
	\end{equation}
We can find that \eqref{equ:semi_algebraic_F1} is semi-algebraic.

Similarly, the feasible domain $\mathcal{F}_2$ defined by \eqref{equ:rn_pairing_relaxation_analyze_convergence} is
	\begin{align}
	\label{equ:semi_algebraic_F2}
%	\begin{aligned}
		\mathcal{F}_2
		=&\bigcup_{\mathbf{U}\in \mathcal{U}}\left.\left\{\left(\mathbf{\bar{p}}, \mathbf{\hat{X}}, \mathbf{X}\right)\right|
		R_{n,n} \! \geqslant\! R_n, \hat{x}_{m,\!n}\!=\!1,\! \hat{x}_{m',\! n'}\!=\!0
		\right\}\nonumber\\
		=&\bigcup_{\mathbf{U}\in\! \mathcal{U}}\left[\left(\bigcap_{(m, \!n)\in \!\mathbf{U}} \right.\left\{\left(\mathbf{\bar{p}},\! \mathbf{\hat{X}}, \!\mathbf{X}\right)\Big|
		\tilde{p}_{m,\!n} \!\leqslant \!0, \hat{x}_{m,\!n}\!=\!1
		\right\}\right)\nonumber\\
		&\!\bigcap
		\left.\!\left(\bigcap_{\left(m',\! n'\right)\notin\! \mathbf{U}}\!\left\{\!\left(\mathbf{\bar{p}},\! \mathbf{\hat{X}},\! \mathbf{X}\right)\Big|
		\tilde{p}_{m',\!n'}\!\leqslant\! 0, \hat{x}_{m',\!n'}\!=\!0\!
		\right\}\!\right)\!\right],
%	\end{aligned}
\end{align}
where $\tilde{p}_{m,n}$ is given by
\begin{equation}
	\tilde{p}_{m,n} = -\left|h_n\right|^2p_n^2+\left(p_m-p_n\right)\sigma^2.
	\end{equation}
We can find that \eqref{equ:semi_algebraic_F2} is also semi-algebraic.

The feasible domain $\mathcal{F}_3$ defined by \eqref{equ:sum_power_pairing_relaxation_analyze_convergence}--\eqref{equ:last_constraints_analyze_convergence} is also semi-algebraic, since these constraints are polynomial.

According to \cite[eq. I.2.9.1]{Shiota1997},
the intersection of semi-algebraic sets is also semi-algebraic. Therefore, the feasible domain of \eqref{equ:newton_pairing_matrix_analyze_convergence}, which is given by
\begin{equation}
	\mathcal{F} = \mathcal{F}_1\cap\mathcal{F}_2\cap \mathcal{F}_2,
\end{equation}
is also semi-algebraic. As a result, $\delta_{\mathcal{F}}\left(\bar{\mathbf{p}}, \mathbf{\hat{X}}, \mathbf{X}\right)$ is semi-algebraic, since the indicator function of a semi-algebraic set is semi-algebraic~\cite{Donatelli2019}.
%	\end{proof}

%According to Lemma \ref{lemma:objective_semi_algebraic} and Lemma \ref{lemma:indicator_semi_algebraic},
As discussed above, both the objective function of \eqref{equ:newton_pairing_matrix_analyze_convergence} and indicator function of feasible domain are semi-algebraic. Since the finite sum of semi-algebraic functions is also semi-algebraic \cite{Givens2012}, the mapping of the overall algorithm, i.e.,  \eqref{equ:objective_overll_algorithm} is semi-algebraic. This proof is complete.

\section{Proof of Theorem \ref{theorem:convergence_rate}}
\label{appendix:proof_theorem2}
According to Lemma \ref{theorem:power_allocation_semialgebraic}, the mapping of the overall algorithm is semi-algebraic. Thus, $Q$ has the K\L{} property \cite{Bolte2007}. Let $\mathbf{\bar{p}}^{(q)}$, $\mathbf{\hat{X}}^{(q)}$  and $\mathbf{X}^{(q)}$ denote the $\mathbf{\bar{p}}$, $\mathbf{\hat{X}}$ and $\mathbf{X}$ generated in the $q$-th iteration, respectively.
$\left\{\mathbf{\bar{p}}^{(q)}\right\}$ is bounded since $\left|h_n\right|$, $\left|h_m\right|$ and $\sigma$ are bounded. Similarly, both $L$ and $s$ are bounded. Therefore, subproblems
		$\min_{\mathbf{\bar{p}}}Q\left(\bar{\mathbf{p}},  \mathbf{\hat{X}}, \mathbf{X}\right)$,
		$\min_{\mathbf{\hat{X}}}Q\left(\bar{\mathbf{p}},  \mathbf{\hat{X}}, \mathbf{X}\right)$,
and
	$\min_{\mathbf{X}}Q\left(\bar{\mathbf{p}},  \mathbf{\hat{X}}, \mathbf{X}\right)$
converge in each iteration of the overall algorithm. On the other hand, it is easy to know that both $\mathbf{\hat{X}}^{(q)}$ and $\mathbf{X}^{(q)}$ are bounded.
When the algorithm sequentially solves
		$\min_{\mathbf{\bar{p}}}Q\left(\bar{\mathbf{p}},  \mathbf{\hat{X}}, \mathbf{X}\right)$,
		$\min_{\mathbf{\hat{X}}}Q\left(\bar{\mathbf{p}},  \mathbf{\hat{X}}, \mathbf{X}\right)$,
and
	$\min_{\mathbf{X}}Q\left(\bar{\mathbf{p}},  \mathbf{\hat{X}}, \mathbf{X}\right)$,
the sequence $\left(\mathbf{\bar{p}}^{(q)}, \mathbf{\hat{X}}^{(q)}, \mathbf{X}^{(q)}\right)$ generated by the algorithm is bounded. According to \cite{Xu2013}, as a bounded sequence generated by the function with K\L{} property, $\left(\mathbf{\bar{p}}^{(q)}, \mathbf{\hat{X}}^{(q)}, \mathbf{X}^{(q)}\right)$ converges to a stationary point of \eqref{equ:newton_pairing_matrix_analyze_convergence}.

Furthermore, let $o^*$ denote the optimum of the algorithm. When $S$ and $L$ exists, it was shown in  \cite{Xu2013,Khan2019} that there exist constant $C$, $\varrho$, and $q_0>0$, satisfying
\begin{equation}
	\left|o_q-o^*\right|  \leqslant  \dfrac{C}{2} q^{-\frac{1}{\varrho}},
	\end{equation}
after $q>q_0$ iterations.
Hence, we have
\begin{equation}
	\left|o_q-o_{q-1}\right| \leqslant \left|o_q-o^*\right|  + \left|o_{q-1}-o^*\right|
	\leqslant C q^{-\frac{1}{\varrho}}.
	\end{equation}
We can choose $\eta$ satisfying
\begin{equation}
	\left|o_q-o_{q-1}\right| \leqslant C q^{-\frac{1}{\varrho}}\leqslant \eta.
	\end{equation}
Hence, we have
\begin{equation}
	q\geqslant \left(\dfrac{C}{\eta}\right)^{\varrho}.
	\end{equation}
The number of iterations of the overall algorithm is given by
\begin{equation}
	\label{equ:number_of_iteration_of_overall_algorithm}
	q \sim  \mathcal{O}\left(\dfrac{1}{\eta^{\varrho}}\right),
	\end{equation}
which concludes this proof.
\end{document}